\documentclass[letterpaper, 11pt]{article}

\usepackage{mathtools}  
\usepackage{amssymb}
\usepackage{amsthm}  

\theoremstyle{plain}
\newtheorem{theorem}{Theorem}
\newtheorem{lemma}[theorem]{Lemma}
\newtheorem{corollary}{Corollary}[theorem]
\newtheorem{prop}[theorem]{Proposition}
\newtheorem{conj}{Conjecture}

\theoremstyle{remark}
\newtheorem*{remark}{Remark}

\usepackage[theoremfont, trueslanted, largesc]{newpxtext}  
\usepackage[vvarbb, noamssymbols]{newpxmath}  

\usepackage[T1]{fontenc}  
\usepackage[english]{babel}  
\AddBabelHook{frspacing}{afterextras}{\frenchspacing}  

\usepackage[hmargin=1.2in, marginparwidth=1in, marginparsep=0.1in, vmargin=1in, headheight=0.6in, headsep=0.3in, footskip=0.5in, heightrounded]{geometry}  

\usepackage{subcaption}
\usepackage[numbers, sort&compress]{natbib}  
\usepackage[tracking]{microtype}  
\usepackage{fnpct}  

\usepackage{graphicx}  
\graphicspath{{./images/}}  

\usepackage{booktabs}

\allowdisplaybreaks[1]  

\renewcommand*{\emptyset}{\varnothing}  

\newcommand*{\Reals}{\mathbb{R}}
\newcommand*{\Integers}{\mathbb{Z}}

\DeclareMathOperator{\E}{\mathbb{E}}

\DeclarePairedDelimiterX\set[1]\lbrace\rbrace{\mkern1.5mu\def\suchthat{\;\delimsize|\;}#1\mkern1.5mu}

\usepackage[colorlinks, allcolors=blue]{hyperref}  
\hypersetup{
	pdftitle=Splitting Planar Necklaces with Few Straight Cuts,
	pdfauthor=Rishi Advani and Abolfazl Asudeh,
	pdfkeywords=Necklace splitting; ham sandwich; fair division,
	pdflang=en-US
}

\usepackage{ellipsis}  

\title{Splitting Planar Necklaces with Few Straight Cuts}
\author{Rishi Advani \and Abolfazl Asudeh}
\date{}

\begin{document}

\maketitle

\begin{abstract}
In the classical necklace splitting problem, beads of different colors lie on a line, and the goal is to divide them among $k$ agents using as few cuts as possible so that every agent receives the same number of beads of each color. We study the planar analog in which a cut is a full straight line and the objective is minimizing the number of lines: $m$ points colored red and blue must be separated by straight-line cuts into regions that can be allocated to $k$ agents so that each agent receives the same number of points of each color. Projection to a line gives $2(k-1)$ cuts, and $k-1$ are sometimes necessary; we give algorithms, static and dynamic, that provably close part of this gap. We develop several approaches that capitalize on the extra degree of freedom: a ``star'' search that produces solutions with as few as $\bigl\lceil \frac{k}{2} \bigr\rceil$ cuts and never more than $k$ whenever a fair contiguous circular partition exists, with further guarantees for evenly spaced input, and a general algorithm that produces at most $k + \log_2(k+1) - 2$ cuts in $\Theta(m \log k)$ time, built on a one-dimensional necklace splitting subroutine that requires no presorting. We then adapt the general algorithm to a dynamic setting, maintaining exact fairness and the cut bound under arbitrary relocation of points, with expected update time that scales with the total distance moved.
\end{abstract}

\section{Introduction}
Fair division of discrete resources is a classical topic at the intersection of combinatorics, topology, and algorithms. A key result is the necklace splitting theorem: a necklace with beads of $n$ colors can always be divided among $k$ agents using at most $n(k-1)$ cuts so that every agent receives the same number of beads of each color, and this bound is tight~\cite{Goldberg1985,Alon1986,Alon1987}; the continuous version goes back to Hobby and Rice~\cite{Hobby1965}. Beyond its combinatorial appeal, necklace splitting has found algorithmic applications in fair resource allocation, including memory- and time-efficient fair data structures~\cite{Shahbazi2024}, and efficient and dynamic algorithms for the one-dimensional problem have recently been developed~\cite{Alon2021,Shahbazi2024,Advani2026a}.

Many resources, however, are not arranged along a line. When the beads are points in the plane --- voters in a territory, facilities in a service area, records embedded in a two-dimensional feature space --- cutting with straight lines is the natural analog; a line is what a cut \emph{is} when cutting is a physical action: a stroke of a knife, saw, or laser across the whole stock, or a splitline in a districting map, is one operation, and the user pays per stroke. The extra geometric freedom then raises the possibility of using far fewer cuts than the one-dimensional bound $2(k-1)$.

In this paper, we study the following planar version of necklace splitting for two colors: cuts are full straight lines, an agent may receive any union of the cells of the resulting arrangement, and the objective is the number of lines. This differs from the well-studied convex equipartition problem~\cite{Bespamyatnikh2000,Sakai2002,Blagojevic2014}, where the objective is convexity of the pieces, and from hyperplane-arrangement partition problems~\cite{Ramos1996,Blagojevic2023,Hubard2019}, where two agents share many colors rather than many agents sharing two. We chart the landscape between the trivial upper bound of $2(k-1)$ cuts obtained by projection and the $k-1$ geometric lower bound, and we give constructive algorithms --- both static and dynamic --- that provably beat the projection bound.

\subsection{Problem formulation}
There are $m$ points in the plane, of which $r$ are colored red and $b$ are colored blue. There are $k \geq 2$ agents. Both $r$ and $b$ are divisible by $k$. We want to cut the plane with straight lines\footnote{See Appendix~\ref{sec:jordan_curve} for a variant where cuts are instead Jordan curves.} as few times as possible so that the resulting regions can be fairly allocated among the agents: each agent should receive the same number of red points and the same number of blue points. If an input point is directly on a cut, an algorithm may choose which side it belongs to. Furthermore, the algorithm is allowed to make coincident cuts.

\subsection{Contributions}
We make the following contributions.
\begin{description}
\item[Landscape of the problem (Section~\ref{sec:initial_exploration})] We observe that projection to one dimension yields $2(k-1)$ cuts, that $k-1$ cuts are sometimes necessary, that recursive ham sandwich cuts give $k-1$ cuts when $k$ is a power of two, and that a natural one-agent-at-a-time approach is blocked by the nonexistence of $\alpha$--ham sandwich cuts.
\item[The star approach (Section~\ref{sec:star_approach})] We cut with concurrent lines through a common viewing point (a $2j$-fan~\cite{Barany2001} whose rays come in antipodal pairs) reducing the planar problem to splitting a circular necklace in which each line supplies two antipodal cut points. We show that the minimum number of lines needed is governed by a maximum matching in an auxiliary graph, and obtain a polynomial-time algorithm using at most $k$ cuts whenever some viewing point admits a solution. For evenly spaced input we prove a parity dichotomy and cut bounds for $m \in \set{2k, 3k, 4k}$ that beat the projection bound.
\item[Three agents (Section~\ref{sec:three_agents})] For $k=3$, three cuts suffice, as follows from existing convex 3-partition results~\cite{Ito2000,Bespamyatnikh2000,Sakai2002}; we give an elementary linear-time algorithm, and we prove that two parallel cuts suffice when $m=6$.
\item[A general algorithm (Section~\ref{sec:general_alg})] We design a one-dimensional necklace splitting subroutine that runs on unsorted input in $\Theta(M \log K)$ time and can split off allocations for groups of agents. Combining it with recursive ham sandwich cuts yields a solution with at most $k + \log_2(k+1) - 2$ cuts computed in $\Theta(m \log k)$ time. For input in strictly convex position, we achieve $k-1$ cuts.
\item[A dynamic algorithm (Section~\ref{sec:dynamic})] We maintain a fair solution under arbitrary relocations of points. Fairness and the cut bound are preserved deterministically for every update, and the worst-case update time matches one static run. Under a stochastic bounded-displacement model, the expected update time scales with the distance moved, for arbitrarily long update sequences; the same bound also holds for a single adversarial batch applied to a random instance. The repair-region selection problem, NP-complete in general~\cite{Advani2026a}, is solved exactly here via enumeration and a parameterized Steiner-forest algorithm.
\item[Structural results (Appendices)] We observe that $k-1$ cuts suffice when cuts are Jordan curves instead of straight lines, and we derive exact (non-asymptotic) counts of the angular orderings of planar point sets, refining the asymptotic bounds of D\'{i}az-B\'{a}\~{n}ez et al.~\cite{DiazBanez2014} and Kyn\v{c}l and Soukup~\cite{Kyncl2025}; this quantity governs the number of distinct necklaces the star approach must test for fairness.
\end{description}

\begin{table}[pt]
\centering
\caption{Summary of results. All cut counts are upper bounds unless otherwise stated. For the dynamic algorithm with random input, running time is in expectation and $\varepsilon_m = O\bigl(\sqrt{(\log m)/m}\bigr)$. \\
$^\dagger$Assumes a viewing point whose induced circular necklace admits a fair contiguous $k$-partition (Section~\ref{sec:star_approach}).}
\label{tab:results}
\small
\setlength{\tabcolsep}{5pt}
\begin{tabular}{lccr}
\toprule
Input & Cuts & Time & Where \\
\midrule
\multicolumn{4}{l}{\textit{Baselines and bounds}} \\
Arbitrary (projection) & $2(k-1)$ & $\Theta(m \log m)$ & \S{} \ref{sec:initial_exploration} \\
Worst case (lower bound) & $\geq k-1$ & --- & Prop.~\ref{prop:lower_bound} \\
\addlinespace
\multicolumn{4}{l}{\textit{Static algorithms}} \\
Arbitrary (star approach)$^\dagger$ & $\lceil k/2 \rceil$ to $k$ & $\Theta(m^5)$ & \S{} \ref{sec:star_approach} \\
Evenly spaced$^\dagger$, $k$ even or $m = 2k$ & $\lceil k/2 \rceil$ & $\Theta(m)$ & Thm.~\ref{thm:parity_dichotomy} \\
Evenly spaced$^\dagger$, $k$ odd and $m > 2k$ & $k$ & $\Theta(m)$ & Thm.~\ref{thm:parity_dichotomy} \\
Evenly spaced, $k$ even & $2(k-1) - 1$ & $\Theta(m)$ & \S{} \ref{sec:star_evenly_spaced} \\
Evenly spaced, $m = 2k$ & $k$ & $\Theta(m)$ & \S{} \ref{sec:star_evenly_spaced} \\
Evenly spaced, $m = 3k$ & $\lfloor (4k-2)/3 \rfloor$ & $\Theta(m)$ & Prop.~\ref{prop:m3k} \\
Evenly spaced, $m = 4k$, $r = 2k$ & $k$ & $\Theta(m)$ & \S{} \ref{sec:star_evenly_spaced} \\
Evenly spaced, $m = 4k$, $r = k$ & $\lfloor 7k/4 \rfloor$ & $\Theta(m)$ & Prop.~\ref{prop:m4k} \\
Arbitrary, $m = 6$ and $k=3$ & $2$ & $\Theta(1)$ & Thm.~\ref{thm:six_points} \\
Arbitrary & $k + \log_2(k+1) - 2$ & $\Theta(m \log k)$ & Thm.~\ref{thm:general_alg} \\
Convex & $k-1$ & $\Theta(m \log k)$ & Prop.~\ref{prop:convex} \\
\addlinespace
\multicolumn{4}{l}{\textit{Dynamic algorithm}} \\
Arbitrary & $k + \log_2(k+1) - 2$ & $\Theta(m \log k)$ & Prop.~\ref{prop:batch_correct} \\
Random, displacement $\leq \delta$ & $k + \log_2(k+1) - 2$ & $O\bigl((\delta + \varepsilon_m)\, m (\log k)^2 + \log k\bigr)$ & Thm.~\ref{thm:single} \\
\bottomrule
\end{tabular}
\end{table}

\subsection{Related work}
Throughout this section, $d$ denotes the dimension, $n$ the number of colors, $k$ the number of agents, and $m$ the number of points. Kano and Urrutia~\cite{Kano2020} survey problems on colored point sets in the plane, including splitting problems, and Rold\'{a}n-Pensado and Sober\'{o}n~\cite{RoldnPensado2021} survey mass partitions broadly.

\subsubsection{Necklaces in one dimension}
Several works refine the necklace splitting theorem. Asada et al.~\cite{Asada2018} study a ``binary'' version in which $k=2^j$ agents are placed at the vertices of a $j$-cube and adjacent pieces must go to adjacent agents; Joji\'{c} et al.~\cite{Joji2021} generalize their results. Gebhardt et al.~\cite{Gebhardt2022} extend the two-agent theorem to necklaces with \emph{signed} beads, where beads count positively or negatively toward an agent's share. These are existence results.

On the algorithmic side, Meunier~\cite{Meunier2008,Meunier2014} gives constructive, combinatorial proofs of the necklace splitting theorem. Alon and Graur~\cite{Alon2021} give polynomial-time approximation algorithms, trading extra cuts or an $\varepsilon$-approximate split for efficiency, in both offline and online models. Advani et al.~\cite{Advani2026a} maintain a fair necklace splitting under relocation, insertion, and deletion of beads, in linear time with the optimal $2(k-1)$ cuts for two colors, and in polylogarithmic time for approximate fairness.

Filos-Ratsikas and Goldberg~\cite{FilosRatsikas2019} show that the necklace splitting and discrete ham sandwich problems are PPA-complete; Filos-Ratsikas et al.~\cite{FilosRatsikas2021} show that necklace splitting with a prime number of thieves is in PPA-$p$. Tractable cases arise from well-separation: Borzechowski et al.~\cite{Borzechowski2025} and Schnider et al.~\cite{Schnider2025} solve the fair and unfair variants of necklace splitting for $k=2$ in polynomial time on $n$-separable necklaces, and Chiu et al.~\cite{Chiu2020} place the related $\alpha$--ham sandwich problem in the UEOPL complexity class.

\subsubsection{Multidimensional necklaces}
De Longueville and \v{Z}ivaljevi\'{c}~\cite{deLongueville2008} split $n$ measures on the $d$-cube among $k$ agents using $n(k-1)$ axis-aligned hyperplanes, with any prescribed number of cuts in each direction. Karasev et al.~\cite{Karasev2016} obtain analogous results for iterated cuts without the requirement that the hyperplanes are axis-aligned. Blagojevi\'{c} and Sober\'{o}n~\cite{Blagojevic2017a} split measures using convex, iterated convex, and iterated hyperplane partitions without fixed directions, and \v{Z}ivaljevi\'{c}~\cite{Zivaljevic2015} use polyhedral curtains. All of these are existence results.

Closest to our setting is the discrete multidimensional necklace, a colored point set in $\Integers^d$ cut with axis-aligned hyperplanes, posed by Grytczuk and Lubawski~\cite{Grytczuk2015}. Unlike in one dimension, it does not reduce to the continuous version: Laso\'{n}~\cite{Lason2015} shows that between $d(n-1)\frac{k}{2}$ and $(2d-1)n(k-1)$ cuts are needed. Our problem is exactly the planar discrete necklace with two colors, except that the cuts are arbitrary lines rather than axis-aligned ones.

\subsubsection{Partitioning using lines and fans}
Partitioning the plane with a single line is the ham sandwich cut, computable in linear time in the plane~\cite{Lo1994}. B\'{a}r\'{a}ny et al.~\cite{Barany2007} show that for well-separated convex bodies a single hyperplane can cut prescribed, possibly unequal, fractions of each; Steiger and Zhao~\cite{Steiger2009} prove the point-set version and compute the cut in $O\bigl(m(\log m)^{d-3}\bigr)$ time for $d \geq 3$, and Bereg~\cite{Bereg2012} in linear time for any fixed $d$. Abbott et al.~\cite{Abbott2009} maintain a ham sandwich cut under insertions and deletions; they also maintain an two-line equipartition of a single point set.

Partitions using \emph{$j$-fans}, $j$ rays from a common apex, were introduced by B\'{a}r\'{a}ny and Matou\v{s}ek~\cite{Barany2001}, who cut two measures into fractions $\frac{2}{5}$, $\frac{1}{5}$, $\frac{1}{5}$, and $\frac{1}{5}$ with a 4-fan; in a follow-up they equipartition two measures with a 4-fan~\cite{Barany2002}. Bereg~\cite{Bereg2005} proves a discrete version for three point sets, with the apex on any prescribed line, with an algorithm running in $O\bigl(m (\log m)^2\bigr)$ time, which is nearly optimal~\cite{Roy2007}. Our star approach (Section~\ref{sec:star_approach}) cuts with $j$ lines through a common point, i.e., a $2j$-fan whose rays come in antipodal pairs, and the two-cut discussion for $k=3$ in Section~\ref{sec:three_agents} concerns 4-fans of this restricted form.

There are also results on partitioning a single measure or point set. B\'{a}r\'{a}ny et al.~\cite{Barany2010} equipartition a convex body with a 3-fan such that the pieces have equal perimeter. Megiddo~\cite{Megiddo1985} finds two lines forming four equal sectors in linear time; Roy and Steiger~\cite{Roy2007} find the same with orthogonal lines in $\Theta(m \log m)$ time. They also find three concurrent lines forming six equal sectors, as guaranteed by Buck and Buck~\cite{Buck1949}, in the same amount of time.

Another line of work focuses on how many hyperplanes are needed to bisect multiple measures. Barba et al.~\cite{Barba2019} bisect four planar masses with two lines, Hubard and Karasev~\cite{Hubard2019} bisect $d\ell$ measures in $\Reals^d$ with $\ell$ hyperplanes for $d$ a power of two, and Blagojevi\'{c} et al.~\cite{Blagojevic2021} extend this to odd $d$ in further dimensions. For point sets, Schnider~\cite{Schnider2022} shows that finding the $\ell$ lines is PPA-complete, and Deligkas et al.~\cite{Deligkas2022} show PPA-hardness even for $\varepsilon$-approximate solutions. These problems share with ours the feature that an agent receives a union of cells of a line arrangement, but they divide many colors between two agents, whereas we divide two colors among many agents.

\subsubsection{Equitable convex and orthogonal partitions}
Equitable convex partitions ask for $k$ convex pieces, each containing a $\frac{1}{k}$ share of each color. For two colors in the plane, their existence was conjectured by Kaneko and Kano~\cite{Kaneko1999}; it was proved independently by Ito et al.~\cite{Ito2000}, Bespamyatnikh et al.~\cite{Bespamyatnikh2000}, and Sakai~\cite{Sakai2002}. In $\Reals^d$, Sober\'{o}n~\cite{Soberon2011} and Karasev et al.~\cite{Karasev2013} equipartition $d$ measures into $k$ convex pieces; Holmsen et al.~\cite{Holmsen2017} adapt the result to point sets when $m=k(d+1)$, and Blagojevi\'{c} et al.~\cite{Blagojevic2017b} do so for general $m$. Bereg~\cite{Bereg2009} equipartitions two point sets into $k$ regions using at most $k-1$ horizontal and $k-1$ vertical segments. Convex and orthogonal pieces are carved out by segments, so these results do not bound the number of full lines; conversely, our allocations need not be convex or even connected.

\section{Initial Bounds and Approaches}
\label{sec:initial_exploration}
In this section, we explore the landscape of the problem and obtain some initial results.
If we project all the input points onto a line, the resulting arrangement of points can be interpreted as an instance of the one-dimensional necklace splitting problem. Running the corresponding algorithm~\cite{Shahbazi2024,Advani2026a} gives us a fair solution with $2(k-1)$ cuts~\cite[Proposition~2]{Alon2021}. However, by doing this, we are not taking advantage of the full flexibility of the plane; we may be missing opportunities to reduce the number of cuts.
On the other extreme, the lazy caterer's sequence limits how efficient cuts can be; we need $\Omega\bigl(\sqrt{k}\bigr)$ cuts. This follows from the fact that an arrangement of $j$ lines can produce at most $O\bigl(j^2\bigr)$ regions, and at least $k$ regions are needed for fair division to be possible.
We can prove a tighter bound still.

\begin{prop}
\label{prop:lower_bound}
For every valid $k$, $r$, and $b$, there is a problem instance that requires $k-1$ cuts.
\end{prop}

\begin{proof}
Place points $q_1, \dots, q_m$ on a circle, in cyclic order, with $q_1, \dots, q_r$ colored red and $q_{r+1}, \dots, q_m$ colored blue; each color class occupies a contiguous arc.

Fix some fair solution. For the red points to be distributed among $k$ agents, there must be at least $k$ intervals (sets of consecutive input points assigned to the same agent) containing red points. Similarly, there must be at least $k$ intervals containing blue points. Due to the contiguous arcs, there at most two intervals that contain both red and blue points; thus, there are at least $2(k-1)$ distinct intervals. Separating $2(k-1)$ intervals requires $2(k-1)$ cut points. A line can intersect a circle in at most two places, so there must be at least $k-1$ cuts.
\end{proof}

If $k$ is a power of two, we can divide and conquer by recursively finding ham sandwich cuts~\cite{Bespamyatnikh2000}. In total, we make $k-1$ cuts, which is significantly better than the projection approach. However, this approach does not naively generalize to other values of $k$. There is no way to inductively solve the problem one agent at a time in $k-1$ cuts; this would imply the existence of an $\alpha$--ham sandwich cut~\cite{Barany2007}, which does not always exist~\cite[Figure~1]{Chiu2020}. Their counterexample can be generalized to any higher dimension for any $\alpha \in \bigl(0, \frac{1}{2}\bigr) \cup \bigl(\frac{1}{2}, 1\bigr)$: have the red points tightly clustered within a blue hypersphere.

\section{Star Approach}
\label{sec:star_approach}
In this section, we describe an approach that beats the $2(k-1)$ bound in certain cases.
Fix a \emph{viewing point} in the plane and project the input points radially onto a circle centered at that point, creating a circular necklace\footnote{We can also apply this approach to higher dimensions by first projecting the points onto a plane.}. We cut with \emph{generalized cuts}: full lines through the viewing point. Each generalized cut meets the circle in two antipodal points, inducing a \emph{primary} and a \emph{secondary} one-dimensional cut on the circular necklace.

The resulting algorithm is as follows. Two arrangements of the viewing point govern the search. Let $\mathcal{A}_{\mathrm{line}}$ be the arrangement of the $\binom{m}{2}$ lines spanned by pairs of input points, and let $\mathcal{A}_{\mathrm{ray}}$ be the coarser arrangement obtained by deleting the open segment between each pair, leaving only the two rays. The circular order of the projected points is constant as the viewing point varies within a cell of $\mathcal{A}_{\mathrm{ray}}$: crossing the open segment between two beads makes them antipodal without changing the cyclic order~\cite{Devillers2007,DiazBanez2014}. Which pairs of gaps can be served by a single line through the viewing point additionally depends on the antipodal structure, which is constant only within a cell of $\mathcal{A}_{\mathrm{line}}$. Both arrangements have $O(m^4)$ cells, and a representative viewing point for each cell is easy to compute. For every cell of $\mathcal{A}_{\mathrm{ray}}$, we (i) decide whether or not its necklace admits a fair contiguous $k$-partition, i.e., a set of cut points splitting the necklace into $k$ arcs, each containing exactly $\frac{r}{k}$ red and $\frac{b}{k}$ blue beads; and, if so, for every cell of $\mathcal{A}_{\mathrm{line}}$ it contains, we (ii) compute the minimum number of generalized cuts realizing such a partition. Section~\ref{sec:star_matching} reduces (i) to an elementary computation and (ii) to a maximum-matching problem. We output the best solution found over all cells (if one exists); the lines through the chosen viewing point then constitute a fair solution of the planar instance. In practice, one may instead sample a subset of cells at random, trading the guarantee of optimality over viewing points for speed.

Before analyzing the algorithm further, we give a baseline lower bound: $j$ concurrent lines create only $2j$ sectors, and $k$ agents require at least $k$ nonempty regions, so no viewing point admits a solution with fewer than $\bigl\lceil \frac{k}{2} \bigr\rceil$ generalized cuts.

\subsection{Maximum matching characterization}
\label{sec:star_matching}
The source of the approach's power is that a cut point is never harmful: an unneeded cut point merely subdivides one agent's arc into two arcs, both of which can be allocated to that agent. Cut points come in rigid antipodal pairs, but only some of them need to be useful; the question is how many generalized cuts can produce two useful cut points. We formalize this as a matching problem.

Fix a viewing point that does not lie on the line through any two input points; the projected beads then occupy $m$ distinct angles. Let the \emph{gaps} be the $m$ open arcs between consecutive beads. A \emph{fair contiguous $k$-partition} $\Pi$ of the necklace is a set of $k$ cut points, one in each of $k$ distinct gaps (the \emph{cut gaps} $G_1, \dots, G_k$ of $\Pi$), such that each of the $k$ arcs between consecutive cut points (the \emph{bead arcs of $\Pi$}) contains exactly $\frac{r}{k}$ red and $\frac{b}{k}$ blue beads. A cut point directly on a bead is assigned to an adjacent gap. Given a cut gap $G_i$, its closure $\overline{G_i}$ is the union of the cut gap with its two bounding beads. Since every bead arc of $\Pi$ contains $\frac{m}{k} \geq 2$ beads, no two cut gaps share a bounding bead; their closures are pairwise disjoint. A set of generalized cuts \emph{realizes} $\Pi$ if the regions they create can be allocated so that each agent receives exactly the beads of one bead arc of $\Pi$.

Define the \emph{antipodal compatibility graph} $H(\Pi)$: its vertices are the cut gaps of $\Pi$, and $G_i G_j$ ($i \neq j$) is an edge if there is an angle $\varphi$ with $\varphi \in \overline{G_i}$ and $\varphi + \pi \in \overline{G_j}$, i.e., a single line can supply cut points to both gaps. Let $M(\Pi)$ be the size of a maximum matching in $H(\Pi)$.

\begin{figure}[pt]
\centering
\includegraphics[page=1]{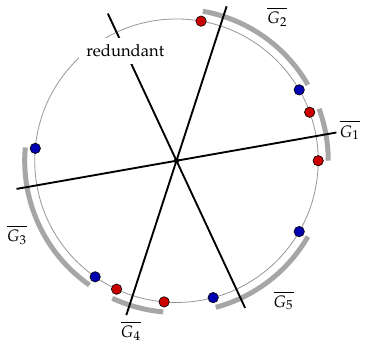}
\hfill
\includegraphics[page=2]{tikz}
\hfill
\includegraphics[page=3]{tikz}
\caption{A visualization of the star approach with $k=5$ and $m=10$. The thick gray arcs outside the circle are the closures $\overline{G_i}$ of the five cut gaps. In the quotient circle $S'$ (Lemma~\ref{lem:H_structure}), $A_1, \dots, A_4$ form a tree, and $A_5$ forms an isolated vertex. The maximum matching in $H(\Pi)$ has size $2$, so $k - M(\Pi) = 3$ lines realize $\Pi$: two lines each serve a matched pair, and the line through $G_5$ has a redundant antipodal cut point inside one agent's bead arc.}
\end{figure}

\begin{theorem}
\label{thm:matching}
The minimum number of generalized cuts realizing $\Pi$ is exactly $k - M(\Pi)$.
\end{theorem}

\begin{proof}
(Upper bound.) Fix a maximum matching. For each matched pair $\set{G_i, G_j}$, place one line at a witnessing angle; its two cut points serve both gaps (resolving any on-bead cut point by the side-choice rule; the beads involved are pairwise distinct because the closures of the cut gaps are pairwise disjoint, so the side assignments do not conflict). For each unmatched cut gap, place one line through its interior; its antipodal cut point is redundant. Every cut gap is served, and we have a fair allocation using $k - M(\Pi)$ lines.

(Lower bound.) Consider any set of lines realizing $\Pi$. Each boundary between two consecutive bead arcs of $\Pi$ must be separated by some cut point lying in the closure of the corresponding cut gap (with side choices consistent with $\Pi$). Charge each cut gap to a line supplying such a point. A line has only two (antipodal) cut points, so it is charged by at most two cut gaps, and if it is charged by two, that pair is by definition an edge of $H(\Pi)$. Since each cut gap is charged to exactly one line, the doubly charged lines induce a matching in $H(\Pi)$, of some size $M' \leq M(\Pi)$, and the number of lines is at least $M' + (k - 2M') = k - M' \geq k - M(\Pi)$.
\end{proof}

\begin{corollary}
\label{cor:star_bounds}
If some viewing point induces a necklace admitting a fair contiguous $k$-partition $\Pi$, then the star approach produces a solution with $k - M(\Pi) \in \bigl[\bigl\lceil \frac{k}{2} \bigr\rceil,\, k\bigr]$ generalized cuts. In particular, it beats the $2(k-1)$ bound for every $k \geq 3$.
\end{corollary}

\begin{proof}
By Theorem~\ref{thm:matching}, the partition $\Pi$ is realized by exactly $k - M(\Pi)$ generalized cuts, and the star approach computes such a solution. A matching in a graph on $k$ vertices has size at most $\bigl\lfloor \frac{k}{2} \bigr\rfloor$, so $0 \leq M(\Pi) \leq \bigl\lfloor \frac{k}{2} \bigr\rfloor$ and $k - M(\Pi) \in \bigl[\bigl\lceil \frac{k}{2} \bigr\rceil, k\bigr]$. Finally, $k < 2(k-1)$ for every $k \geq 3$.
\end{proof}

Both extremes of the range occur. Suppose the projected points are contained in an open semicircle. All gaps except the \emph{complementary gap} --- the one spanning the unoccupied side --- lie inside that semicircle, which contains no antipodal pair of points, so every edge of $H(\Pi)$ is incident to the complementary gap and $M(\Pi) \leq 1$; if the complementary gap is not a cut gap of $\Pi$, then $H(\Pi)$ has no edges, $M(\Pi) = 0$, and realizing $\Pi$ costs the full $k$ lines. In practice, points will be spread more evenly around the viewing point, and Section~\ref{sec:star_evenly_spaced} shows that even spacing forces $M(\Pi) = \bigl\lfloor \frac{k}{2} \bigr\rfloor$ for even $k$.

We now analyze the running time of the star approach. First, we develop an efficient algorithm to compute a maximum matching in $H(\Pi)$; if a given necklace admits a fair contiguous $k$-partition, we run the algorithm. Next, we discuss the time complexity of the search for such necklaces over possible viewing points.

\begin{lemma}[Structure of the compatibility graph]
\label{lem:H_structure}
$H(\Pi)$ is a forest, or consists of a single cycle together with pendant vertices, each adjacent to exactly one vertex of the cycle. In particular, $H(\Pi)$ has at most $k$ edges and at most one cycle.
\end{lemma}

\begin{proof}
Let $q \colon S \to S'$ be the quotient of the circle $S$ that identifies antipodal points; $S'$ is a circle of half the circumference, and $A_i := q(\overline{G_i})$ is a closed arc of $S'$ (possibly all of $S'$). By definition of $H(\Pi)$, $G_i G_j$ is an edge if and only if $A_i \cap A_j \neq \emptyset$. Every point of $S'$ has exactly two $q$-preimages, and each preimage lies in at most one $\overline{G_i}$, since these closures are pairwise disjoint; hence no point of $S'$ lies in three of the arcs $A_1, \dots, A_k$. We refer to this as the ply bound.

If some $A_i$ is all of $S'$, the ply bound makes the remaining arcs pairwise disjoint, so $H(\Pi)$ is a star, which is a tree. Assume henceforth that every $A_i$ is a proper arc.

First, we show that the arcs of any cycle of $H(\Pi)$ (the images under $q$ of each vertex in the cycle) cover $S'$. Suppose the arcs of a cycle $C$ miss a point $p$; cutting $S'$ at $p$ turns them into intervals on a line. Any three pairwise intersecting intervals have a common point, which the ply bound forbids, so the intersection graph of these intervals is triangle-free; it is also an interval graph, hence chordal, and a triangle-free chordal graph is acyclic, which is a contradiction.

Now suppose $H(\Pi)$ contains a cycle, and let $C$ be a shortest one; $C$ is chordless, since a chord would yield a shorter cycle. Let $D_1$ be the set of points of $S'$ covered by exactly one arc of $C$, and $D_2$ the set covered by two or more; by the previous paragraph, $D_1 \cup D_2 = S'$. Each arc consisting of points in $D_1$ must lie within a single arc of $C$; if it spanned multiple arcs, it would contain the shared endpoint of two arcs, which would make it belong to $D_2$. Let $B$ be any arc not in $C$. A point of $B \cap D_2$ would lie in three arcs, violating the ply bound; hence $B \subseteq D_1$, and $B$ is contained within a single arc $A$ of $C$. Thus $B$ intersects $A$ and no other arc of $C$. Finally, two arcs $B, B'$ not in $C$ are disjoint: a common point would lie in $B$, in $B'$, and in an arc of $C$, again violating the ply bound. Hence every vertex outside $C$ is a pendant vertex attached to exactly one vertex of $C$. If instead $H(\Pi)$ contains no cycle, it is a forest.
\end{proof}

\begin{prop}
\label{prop:matching_greedy}
Given the cut gaps in circular order, a maximum matching in $H(\Pi)$ can be computed in $\Theta(k)$ time.
\end{prop}

\begin{proof}
The endpoints of $A_1, \dots, A_k$ in circular order around $S'$ are obtained by merging the (already sorted) endpoint lists of the two half-circles in $\Theta(k)$ time; by the ply bound of Lemma~\ref{lem:H_structure}, a sweep over them maintains at most two active arcs at any moment and outputs the edge list of $H(\Pi)$, which has at most $k$ edges. On this graph, run the following greedy algorithm: while some vertex $u$ has degree one, add its unique edge $uv$ to the matching and delete $u$ and $v$; once no such vertex remains, the surviving components are cycles and isolated vertices, and a cycle of length $t$ contributes $\bigl\lfloor \frac{t}{2} \bigr\rfloor$ alternate edges. The greedy algorithm is optimal: if $u$ is a leaf with neighbor $v$, some maximum matching contains $uv$ (a maximum matching cannot leave both $u$ and $v$ unmatched, and if it matches $v$ by another edge, exchanging that edge for $uv$ preserves the size), and deleting $u$ and $v$ reduces the problem to the smaller graph. Since $H(\Pi)$ has at most $k$ edges, the whole procedure runs in $\Theta(k)$ time.
\end{proof}

The search for an optimal viewing point is polynomial. Stage (i) runs once for each of the $O\bigl(m^4\bigr)$ cells of $\mathcal{A}_{\mathrm{ray}}$\footnote{Exact counts of these cells are derived in Appendix~\ref{sec:angular_ordering}, and since only the induced coloring matters, the effective number is smaller still~\cite{DiazBanez2014}.}. Stage (ii) runs once for each of the $O\bigl(m^4\bigr)$ cells of $\mathcal{A}_{\mathrm{line}}$. For a fixed necklace, the cut gaps of any fair contiguous partition sit at circular bead-index spacing exactly $\frac{m}{k}$, so there are at most $\frac{m}{k}$ candidate partitions, and with prefix sums over the circular color sequence, all fair ones are identified in $\Theta(m)$ total time. For each fair partition, Proposition~\ref{prop:matching_greedy} computes an optimal realization in $\Theta(k)$ time, so a cell is fully processed in $\Theta(m)$ time once its bead order is known. Sweeping the cells of $\mathcal{A}_{\mathrm{line}}$ and updating the ordering incrementally~\cite{Devillers2007} avoids the logarithmic factor from sorting, giving a total running time of $O\bigl(m^5\bigr)$.

\subsection{Evenly spaced input}
\label{sec:star_evenly_spaced}
In this section, we assume evenly spaced input to make the problem more tractable. We measure angular position in \emph{bead units} of $\frac{2\pi}{m}$, so that beads sit at integers modulo $m$ and gaps are the open unit intervals between them; the antipodal map is translation by $\frac{m}{2}$. Call a family of $k$ gaps whose indices form a single residue class modulo $\frac{m}{k}$ a \emph{lattice}, and its members \emph{lattice gaps}. The cut gaps of a fair contiguous $k$-partition always form a lattice, since consecutive arcs contain exactly $\frac{m}{k}$ beads.

\begin{lemma}[Antipodal offset]
\label{lem:offset}
Let $L$ be a lattice of gaps of an evenly spaced necklace. If $k$ is even, the antipodal image of every lattice gap is again a lattice gap. If $k$ is odd, the center of the antipodal image of a lattice gap lies at distance exactly $\frac{m}{2k} \geq 1$ (in bead units) from the nearest lattice gap center; in particular, the closures of the lattice gaps and of their antipodal images are disjoint when $m > 2k$ and meet only at beads when $m = 2k$.
\end{lemma}

\begin{proof}
The lattice gap centers form an arithmetic progression with difference $\frac{m}{k}$, and the antipodal map translates it by $\frac{m}{2} = \frac{m}{k} \cdot \frac{k}{2}$. If $k$ is even, this is a multiple of $\frac{m}{k}$, so the progression is mapped to itself. If $k$ is odd, the residue of $\frac{m}{2}$ modulo $\frac{m}{k}$ is $\frac{m}{2k}$, which is at least $1$ since $m = r + b \geq 2k$; two closed unit intervals whose centers are at distance $\frac{m}{2k}$ intersect only if $\frac{m}{2k} \leq 1$, i.e., $m = 2k$, in which case they share a single endpoint, a bead.
\end{proof}

\begin{theorem}[Parity dichotomy]
\label{thm:parity_dichotomy}
Let the input be evenly spaced. If $k$ is odd and $m > 2k$, then realizing a fair contiguous $k$-partition (if one exists) requires exactly $k$ generalized cuts; else, it requires exactly $\bigl\lceil \frac{k}{2} \bigr\rceil$.
\end{theorem}

\begin{proof}
If $k$ is even and a fair contiguous $k$-partition $\Pi$ exists, then by Lemma~\ref{lem:offset}, applied to the lattice formed by its cut gaps, the antipodal map pairs lattice gap $i$ exactly with lattice gap $i + \frac{k}{2}$ (indices modulo $k$); these pairs form a perfect matching in $H(\Pi)$, so $M(\Pi) = \frac{k}{2}$ and Theorem~\ref{thm:matching} gives $\frac{k}{2}$ lines.
If $k$ is odd and $m > 2k$, by Lemma~\ref{lem:offset}, no line can serve two lattice gaps, so $H(\Pi)$ has no edges, $M(\Pi) = 0$, and Theorem~\ref{thm:matching} gives exactly $k$ lines for any fair contiguous partition $\Pi$.

Assume $k$ is odd and $m = 2k$. Assume without loss of generality that $(0,1)$ is a lattice gap; then the lattice gaps are $\set{2i \suchthat i \in \Integers_k}$, where gap $j$ denotes the arc between beads $j$ and $j+1$. For $i \in \Integers_k$, the line through the antipodal beads $2i + 1$ and $2i + 1 + k$ has its two cut points on those beads; assigning each such bead to the appropriate side (as the problem formulation permits), the first cut point serves lattice gap $2i$ and the second serves lattice gap $2i + 1 + k \equiv 2\bigl(i + \frac{k+1}{2}\bigr) \pmod{2k}$. Thus each such line serves the pair of lattice gaps with indices $\set{i, i + s}$ in $\Integers_k$, where $s = \frac{k+1}{2}$. Since any common divisor of $s$ and $k$ divides $2s - k = 1$, $s$ and $k$ are coprime, so the pairs $\set{i, i+s}$ form a single cycle of length $k$ on $\Integers_k$, which contains a matching of size $\frac{k-1}{2}$. Choosing such a matching serves $k - 1$ of the lattice gaps with $\frac{k-1}{2}$ lines; one final line through the interior of the remaining lattice gap serves it, with a redundant antipodal cut point. The total is $\frac{k+1}{2} = \bigl\lceil \frac{k}{2} \bigr\rceil$.
\end{proof}

For even $k$, even when no fair contiguous $k$-partition exists, we can guarantee an improvement over the projection bound: find a single generalized cut that divides the necklace fairly in half (it exists by the discrete intermediate value theorem) and then run the one-dimensional necklace splitting algorithm on each half, for a total of at most $2(k-1) - 1$ generalized cuts.

\subsubsection{Constant number of agents}
\label{sec:constant}
By exhaustive computational search~\cite{Advani2026b}, we have found that, for $k=3$ and $m=18$, there are four distinct necklaces (up to rotation, reflection, and inversion of color) that cannot be fairly split using fewer than $2(k-1)$ generalized cuts. For $m=24$, there are 92 distinct necklaces. The specific necklaces are shown in Appendix~\ref{sec:necklaces}.

There are roughly $\frac{2^{m-2}}{km}$ valid necklaces: the numerator is the number of binary strings of length $m$, halved twice to account for reflection and inversion of color; the denominator accounts for the divisibility constraint and rotation. As such, the fraction of necklaces that cannot be split with fewer than $2(k-1)$ generalized cuts is very low.

\begin{conj}
With evenly spaced input, if $k=3$, the fraction of possible problem instances that require $2(k-1)$ generalized cuts goes to zero as $m$ goes to infinity.
\end{conj}

For $k=5$ and $m \leq 20$, there are no instances that require $2(k-1)$ generalized cuts --- at most six are needed (exhaustive computational search).

\subsubsection{Number of beads linear in number of agents}
We can still make certain guarantees even when a necklace does not necessarily admit a fair contiguous $k$-partition; we may allocate noncontiguous sets of beads to agents.
If $m=2k$, then $k$ generalized cuts suffice: place a cut point at the center of each of the $2k$ gaps; antipodal centers pair into $k$ lines, and each of the $2k$ resulting arcs contains exactly one bead, so the arcs can be assigned as desired. For $k>2$, this is better than the $2(k-1)$ bound.
Similarly, if $m=4k$ with an equal number of red and blue beads, then $k$ generalized cuts suffice: place a cut point at $\frac{1}{2} + 2j$ for $j \in \Integers_{2k}$ (all in gaps); antipodal positions differ by $\frac{m}{2} = 2k$, an exact multiple of the spacing, so they again pair into $k$ lines. Each of the $2k$ arcs has two beads: both red, both blue, or one of each. The number of $RR$ and $BB$ arcs must be equal, so we can assign each agent either a pair of $RR$ and $BB$ arcs or two mixed arcs.

\begin{prop}
\label{prop:m3k}
Let the input be evenly spaced with $m = 3k$. Then $k + \bigl\lfloor \frac{T}{3} \bigr\rfloor$ generalized cuts suffice, where $T$ is the number of minority-color triples. Since $T \leq k - 2$, this bound is at most $\bigl\lfloor \frac{4k-2}{3} \bigr\rfloor$, and hence beats the $2(k-1)$ bound for every $k > 2$.
\end{prop}

\begin{proof}
There are either $k$ red and $2k$ blue beads, or $2k$ red and $k$ blue beads. By symmetry, without loss of generality, we can assume we are in the former case. Throughout, we use that $m = 3k \equiv 0 \pmod 3$, so the residue of a bead index modulo $3$ is well defined around the necklace. Let $\ell_1, \dots, \ell_c$ be the lengths of the \emph{maximal red runs} (maximal sets of consecutive red beads).

There are three lattices of gaps with spacing three, one per residue of the gap index modulo $3$; fix one of these lattices. If $k$ is odd, place $k$ primary cut points at the centers of the lattice's gaps, creating $k$ intervals of three beads each. As computed in Lemma~\ref{lem:offset}, the secondary cut point of each line falls at offset $\frac{m}{2k} = \frac{3}{2}$ from a lattice gap center --- exactly on the middle bead of an interval --- and the map from lines to intervals is a bijection, so each interval contains exactly one secondary cut, on its middle bead. We can place each secondary cut on either side of its bead, splitting that interval $1$--$2$ or $2$--$1$ at our choice. If $k$ is even, make $\frac{k}{2}$ generalized cuts through the lattice gaps and $\frac{k}{2}$ through the interiors of the intervals. The resulting construction is equivalent: there are $k$ cut points in lattice gaps and $k$ on beads within intervals; the only difference between the odd and even cases are which cuts are primary and which are secondary.

Let $r_i$ denote the number of intervals with exactly $i$ red beads. Counting intervals and red beads gives $r_0 = r_2 + 2r_3$: we have $\sum_i r_i = k$ and $\sum_i ir_i = k$, and subtracting yields the claim. We allocate as follows.
\begin{description}
\item[One red bead] Each interval is assigned whole to an agent; its interior cut is ignored.
\item[Two red beads] Pair each such interval ($RRB$, $RBR$, or $BRR$) with a distinct $BBB$ interval. In each pattern, some placement of the interior cut carves off a single red bead ($R \mid RB$, $R \mid BR$, $BR \mid R$); the interior cut of the $BBB$ interval carves off a single blue bead ($B \mid BB$). One agent receives the two-bead remainder of the red-heavy interval plus the carved blue bead; another receives the two-bead remainder of the $BBB$ interval plus the carved red bead.
\item[Three red beads] Pair each $RRR$ interval with two distinct $BBB$ intervals. The interior cut carves off one red bead, and one additional interior cut separates the remaining two red beads, i.e., one extra generalized cut per $RRR$ interval (its antipodal cut point is redundant); the two $BBB$ intervals are handled by their own interior cuts as before, yielding three agents with one red and two blue beads each.
\end{description}
The identity $r_0 = r_2 + 2r_3$ shows that the $BBB$ intervals suffice exactly for this pairing, so the chosen lattice is realized with $k + r_3$ generalized cuts.

All that remains is to choose the lattice well. The three lattices partition the beads into consecutive triples with phases $0$, $1$, and $2$ modulo $3$, so every red triple is an interval in exactly one lattice; hence the three values of $r_3$ sum to $T$, and the best lattice has $r_3 \leq \bigl\lfloor \frac{T}{3} \bigr\rfloor$, for a total of $k + \bigl\lfloor \frac{T}{3} \bigr\rfloor$ generalized cuts. Finally, $T \leq k - 2$: if $T = 0$ this is immediate from $k \geq 2$, and otherwise some run has length $\ell_1 \geq 3$ and contributes $\ell_1 - 2$ triples, while every other run contributes $\max(\ell_i - 2, 0) \leq \ell_i$, so $T \leq (\ell_1 - 2) + \sum_{i \geq 2} \ell_i = k - 2$.
\end{proof}

\begin{prop}
\label{prop:m4k}
Let the input be evenly spaced with $m = 4k$ and $k$ red beads. Then $\bigl\lfloor \frac{7k}{4} \bigr\rfloor$ generalized cuts suffice. This beats the $2(k-1)$ bound for $k > 8$.
\end{prop}

\begin{proof}
Place cut points at the $2k$ positions $\frac{1}{2} + 2j$, pairing into $k$ generalized cuts and creating $2k$ two-bead arcs. For each $RB$ arc, there is a $BB$ arc, and the pair can be allocated to an agent without extra cuts. For each $RR$ arc, there are three corresponding $BB$ arcs; the $RR$ arc and one $BB$ arc each need one interior cut for these arcs to be allocated to two agents. Naively this could lead to $2k$ total generalized cuts, but with careful accounting of secondary cuts, we can show $\bigl\lfloor \frac{7k}{4} \bigr\rfloor$ generalized cuts suffice.

For each extra cut with its primary in an $RR$ arc, its secondary falls exactly $\frac{m}{2} = 2k$ positions away --- precisely $k$ arcs away, at the same interior offset. If the secondary lands in a $BB$ arc, that arc no longer needs its own cut, saving a generalized cut. If it lands in an $RB$ arc, there is no immediate benefit (the cut point is redundant and harmless). If it lands in another $RR$ arc, the two arcs share the generalized cut, saving half a cut. In the worst case, (almost) all $RR$ arcs are antipodal to other $RR$ arcs: if $k$ is divisible by $4$, there can be $\frac{k}{2}$ $RR$ arcs, requiring $\frac{k}{4}$ shared cuts for the $RR$ interiors and $\frac{k}{2}$ further cuts for the $BB$ partners, for $\frac{3k}{4}$ extra generalized cuts; for the other residues of $k \bmod 4$, the maximum is likewise $\bigl\lfloor \frac{3k}{4} \bigr\rfloor$. Thus $\bigl\lfloor \frac{7k}{4} \bigr\rfloor$ generalized cuts in total suffice, which is better than the $2(k-1)$ bound for $k > 8$.
\end{proof}

We have already shown that instances requiring the full $2(k-1)$ generalized cuts exist for $k=3$ with $m=6k$ and $m=8k$ (Section~\ref{sec:constant}), but we do not have positive or negative results for $5k$, $7k$, or larger multiples.

\section{Three Agents}
\label{sec:three_agents}
In this section, we develop algorithms that work on arbitrary (even adversarial) input. For $k=2$, the ham sandwich theorem provides a solution, so we focus on $k=3$. Existence of a three-cut solution is not new: the equitable convex partition theorem~\cite{Bespamyatnikh2000,Sakai2002,Ito2000} partitions points in general position into three convex regions with $\frac{r}{3}$ red and $\frac{b}{3}$ blue points each, and the boundaries can be extended to full lines to fit our formulation. Our contribution is an elementary algorithm that runs in linear time and does not require general position.

\begin{prop}
If $k=3$, there is always an exact solution consisting of three cuts.
\end{prop}

\begin{proof}
If we project the points onto a line, we can run the one-dimensional necklace splitting algorithm to find a solution with at most four cuts. If we instead run a single iteration of the one-dimensional algorithm, we now have one agent's allocation with at most two cuts. We can then find a standard ham sandwich cut for the remaining points to allocate between the remaining two agents. Overall, this gives us a solution with at most three cuts.
\end{proof}

While this algorithm already beats the $2(k-1)$ bound, we have experimental evidence~\cite{Advani2026b} that even two cuts may be sufficient. For points uniformly distributed over the lattice points in $[-50,50]^2$ and point sets where the red points are tightly clustered within the same bounds, up through $m=24$, with $2000$ trials for each $m$, no problem instances have been discovered that required more than two cuts.
If we require one cut to be a $\frac{1}{3}$-ham sandwich cut, it is known that there are instances that fail with tightly clustered points~\cite[Figure~1]{Chiu2020}. Analogously, if we enforce that an agent gets two opposite sectors allocated to them\footnote{A 4-fan formed by two lines, with the shares $\frac{1}{3}$, $\frac{1}{3}$, $\frac{1}{6}$, and $\frac{1}{6}$ prescribed, in contrast to the equipartitions by 4-fans of B\'{a}r\'{a}ny and Matou\v{s}ek~\cite{Barany2001,Barany2002}.}, there are instances that do not have a two-cut solution; a smallest possible\footnote{There can be no counterexample with $m=6$ because, by Theorem~\ref{thm:six_points}, there is always a solution with two parallel cuts (a degenerate case of opposite sectors).} instance is shown in Figure~\ref{fig:opposite_counterexample_m9}. Thus, any successful algorithm could not rely exclusively on either case.

\begin{figure}[bhpt]
\centering
\includegraphics[width=0.5\linewidth]{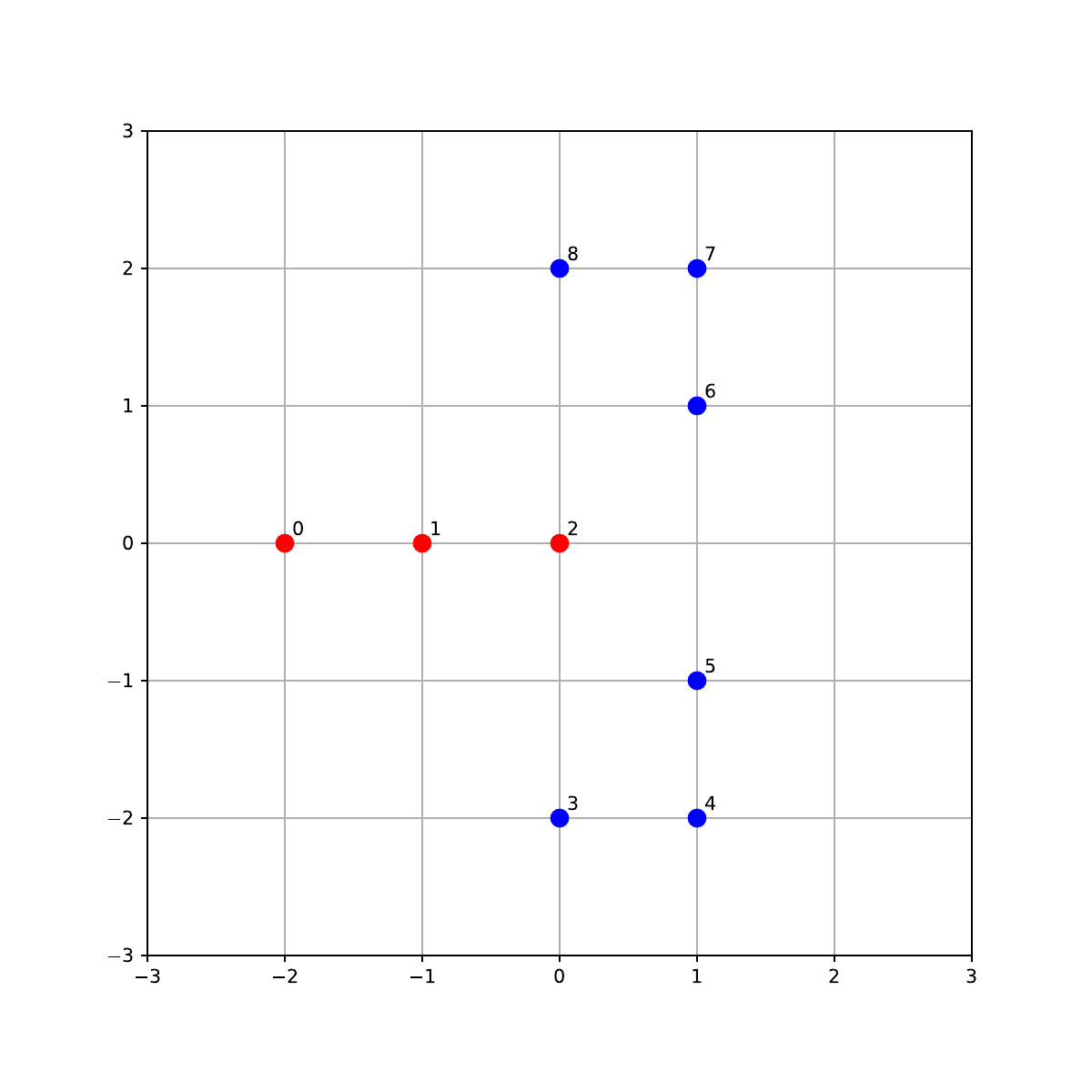}
\caption{An instance with $m=9$ that cannot be fairly split using only two cuts if one agent is required to be allocated opposite sectors.}
\label{fig:opposite_counterexample_m9}
\end{figure}

\begin{conj}
If $k=3$, there is always an exact solution consisting of two cuts.
\end{conj}

\subsection{Six points}
In this section, we focus on the special case where we have exactly three red and three blue points.

\begin{theorem}
\label{thm:six_points}
If $k=3$, $r=3$, and $b=3$, there is always an exact solution consisting of two parallel cuts, and it can be computed in constant time.
\end{theorem}

\begin{proof}
Given an angle $\theta$, project all six input points onto a vector with direction $\theta$. Let $I(\theta)$ be the sum of the indices of the red points in the ordering. E.g., if the first three points are red, we have $I(\theta) = 1 + 2 + 3 = 6$. Note that if each of the sets of indices $\set{1, 2}$, $\set{3, 4}$, and $\set{5, 6}$ includes one red and one blue point, then we have a fair solution with two cuts perpendicular to the vector: a cut between positions 2 and 3, and a cut between positions 4 and 5. Next, note that reversing the direction of the vector switches the positions of red and blue points, so $I(\theta) + I(\theta + \pi) = \sum_{i=1}^6 i = 21$. Finally, note that no 3-subset of $\set{1, \dots, 6}$ that sums to $10$ or $11$ can contain two points from a single pair; if $I(\theta) \in \set{10, 11}$, the instance must have a solution.

Choose some initial angle $\theta_0$. Assume without loss of generality that $I(\theta_0) < \frac{21}{2}$. As the angle $\theta$ continuously varies from $\theta_0$ to $\theta_0 + \pi$, by the discrete intermediate value theorem, $I(\theta)$ must equal $10$ at some point\footnote{Without a general-position assumption, several disjoint pairs of points (or a set of three or more collinear points) may attain equal projections at the same critical angle, so $I(\theta)$ may change by more than one as $\theta$ passes it. In that case, resolve the coincidences one adjacent transposition at a time: each transposition changes $I$ by at most one, so some partial resolution attains the target value.}. Thus, the instance has a solution. Since we made no additional assumptions on the nature of the input, it follows that all instances with $m=6$ have a solution consisting of two parallel cuts.

For the running time: as $\theta$ varies over a half-turn $[\theta_0, \theta_0 + \pi]$, the ordering of the projected points can only change at the critical angles at which two projections coincide, i.e., at the directions perpendicular to the $\binom{6}{2} = 15$ lines spanned by pairs of input points. These critical angles partition the half-turn into at most $15$ arcs, on each of which $I(\theta)$ is constant. We compute the critical angles, evaluate $I(\theta)$ once per arc, and return the two cuts for any arc with $I(\theta) \in \set{10, 11}$, which exists by the argument above. All quantities involved are of constant size, so the algorithm runs in constant time.
\end{proof}

\section{General Algorithm}
\label{sec:general_alg}
We now generalize the algorithm from Section~\ref{sec:three_agents} to $k>3$. First, we present a subroutine for one-dimensional necklace splitting that can produce allocations for groups of agents and does not require presorted input.

\subsection{Unsorted one-dimensional necklace splitting}
\label{sec:unsorted}
Throughout this section, we focus on a single iteration of the subroutine. Let $P$ be the set of $M$ currently unallocated beads, of which $R$ are red and $B$ are blue, let $K \geq 2$ be the number of agents remaining, and let $a < K$ be the number of agents we are allocating beads to in this iteration. By the divisibility assumption, $\beta \coloneqq \frac{M}{K}$ and $Q \coloneqq \frac{aR}{K}$ are integers. Let $s = a\beta$. We think of the necklace as circular; all relevant arithmetic in this section is cyclic, i.e., modulo $M$, and for $i \in \Integers_M$ we write $W(i) \coloneqq \set{i, i+1, \dots, i+s-1} \subseteq \Integers_M$ for the \emph{window} of $s$ consecutive beads starting at $i$, and $f(i)$ for the number of red beads located in $W(i)$.

The goal of the iteration is to compute a window containing exactly $Q$ red beads; since $|W(i)| = s$ is fixed, such a window automatically contains exactly $s - Q = \frac{aB}{K}$ blue beads. Its existence is a proportional refinement of an argument of Alon and Graur~\cite[Proposition~2]{Alon2021}.

\begin{lemma}[Existence]
\label{lem:existence}
There exists $i \in \Integers_M$ with $f(i) = Q$.
\end{lemma}

\begin{proof}
Each bead of $P$ lies in exactly $s$ of the $M$ windows, so $\sum_{i \in \Integers_M} f(i) = sR$, and $\frac{sR}{M} = \frac{aR}{K} = Q$, so the average of $f$ over $\Integers_M$ is exactly $Q$. Consequently either $f = Q$ over all input, or $f$ attains values both strictly above and strictly below $Q$. As a window moves one bead at a time, the number of red beads in it changes by at most one. Thus, by the discrete intermediate value theorem, $f(i)$ must be equal to $Q$ for some $i \in \Integers_M$.
\end{proof}

The subroutine locates such a window in three stages: a coarse partition of $P$ into $K$ rank-blocks (Lemma~\ref{lem:multiselect}); an evaluation of $f$ at the $K$ block boundaries, which either shows the window is exactly aligned with a block or brackets it between two block boundaries (Lemma~\ref{lem:bracket}); and a prune-and-search refinement inside the bracket (Lemma~\ref{lem:refine}). A sample execution is shown in Figure~\ref{fig:subroutine}.

\begin{figure}[pt]
\centering
\includegraphics[page=5]{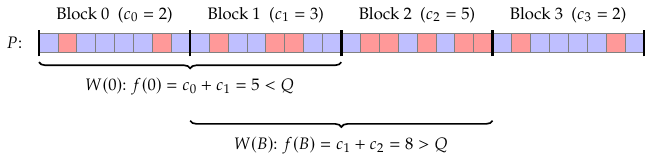}

\bigskip

\includegraphics[page=6]{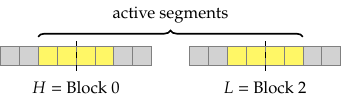}
\caption{A visualization of the unsorted one-dimensional subroutine with $R=12$, $B=20$, $K=4$, $a=2$, and $\beta=8$. Top: after partitioning (Lemma~\ref{lem:multiselect}), the array holds the $K$ rank-blocks in order, but the beads within a block are unsorted; block-aligned windows span $a$ consecutive blocks, and evaluating $f$ at the $K$ block boundaries (Lemma~\ref{lem:bracket}) brackets a window with exactly $Q = \frac{aR}{K} = 6$ red beads between $W(0)$ and $W(\beta)$, so the head block is $H=$ Block $q=$ Block $0$ and the tail block is $L=$ Block $q+a=$ Block $2$. Bottom: by recursive bisection of the active segments (Lemma~\ref{lem:refine}), the desired window will be found.}
\label{fig:subroutine}
\end{figure}

\begin{lemma}[Partition]
\label{lem:multiselect}
In $\Theta(M \log K)$ time, we can rearrange $P$ in an array so that for each $q \in \set{0, \dots, K-1}$ the beads with positions in $[q\beta, (q+1)\beta)$ occupy positions $[q\beta, (q+1)\beta)$ of the array, and compute the red count $c_q$ of each block.
\end{lemma}

\begin{proof}
A linear-time selection algorithm~\cite{Blum1973}, run on an array of elements with target rank $t$, returns the element $x$ of rank $t$, and in linear time, we can rearrange the array so that elements of rank less than $t$ precede $x$ and elements of larger rank follow it. Select the boundary rank $\bigl\lceil \frac{K}{2} \bigr\rceil \beta$ in the full array in $\Theta(M)$ time, then recurse on the two sides, which require $\bigl\lceil \frac{K}{2} \bigr\rceil$ and $\bigl\lfloor \frac{K}{2} \bigr\rfloor$ blocks respectively. The recursion tree has depth $\lceil \log_2 K \rceil$, and the subproblems at each depth are disjoint subarrays of total size at most $M$, so each level costs $\Theta(M)$. A final linear sweep records the counts $c_q$. Thus, the first stage takes $\Theta(M \log K)$ time.
\end{proof}

\begin{lemma}[Boundary evaluation]
\label{lem:bracket}
Given the output of Lemma~\ref{lem:multiselect}, in $\Theta(K)$ time, we can either find $i \in \Integers_M$ with $f(i) = Q$, or find $q \in \Integers_K$ such that $f(q\beta)$ and $f((q+1)\beta)$ lie strictly on opposite sides of $Q$.
\end{lemma}

\begin{proof}
Since $s = a\beta$ is a multiple of the block size, the window starting at a block boundary is exactly $a$ consecutive blocks: $f(q\beta) = c_q + c_{q+1} + \dots + c_{q+a-1}$ (indices mod $K$). Since we have already computed the red-bead counts for each block, we can sweep through the block-aligned windows and compute all $K$ values $f(0), f(\beta), \dots, f((K-1)\beta)$ in $\Theta(K)$ time. Each block appears in exactly $a$ of these $K$ windows, so $\sum_{q \in \Integers_K} f(q\beta) = a \sum_{q \in \Integers_K} c_q = aR = KQ$, and the average of the sampled values is exactly $Q$. Hence either some sample equals $Q$, and we are done, or there is a pair of consecutive samples strictly above and below $Q$. The start of the desired window is thus bracketed within an interval of length $\beta$.
\end{proof}

\begin{lemma}[Refinement]
\label{lem:refine}
Given the bracket of Lemma~\ref{lem:bracket}, in $\Theta(\beta)$ time, we can find $d \in \set{1, \dots, \beta-1}$ with $f(q\beta + d) = Q$, together with the two beads with positions $q\beta + d - 1$ and $q\beta + d$ and the two beads with positions $q\beta + s + d - 1$ and $q\beta + s + d$.
\end{lemma}

\begin{proof}
Assume without loss of generality that $f(q\beta) < Q < f((q+1)\beta)$. For $d \in \set{0, \dots, \beta}$, the window $W(q\beta + d)$ is obtained from $W(q\beta)$ by removing the $d$ leftmost beads of the \emph{head block} $H$ (Block $q$, positions $[q\beta, (q+1)\beta)$) and adjoining the $d$ leftmost beads of the \emph{tail block} $L$ (Block $q + a \bmod K$, positions $[(q+a)\beta, (q+a+1)\beta)$). Let $r_H(d)$ and $r_L(d)$ denote the number of red beads among the $d$ leftmost beads of $H$ and of $L$ respectively, and set $\Delta(d) \coloneqq r_L(d) - r_H(d)$, so that $f(q\beta + d) = f(q\beta) + \Delta(d)$ and the target value for $\Delta(d)$ is $\tau \coloneqq Q - f(q\beta) \geq 1$. Note that $\Delta(0) = 0 < \tau$ and $\Delta(\beta) = f((q+1)\beta) - f(q\beta) > \tau$. As we slide the window one bead at a time, the value of $\Delta(d)$ changes by at most one, so by the discrete intermediate value theorem, $\Delta(d) = \tau$ for some $d$.

We locate the target value by recursive bisection, maintaining $0 \leq \mathit{lo} < \mathit{hi} \leq \beta$ with the numeric invariant $\Delta(\mathit{lo}) < \tau \leq \Delta(\mathit{hi})$, together with the following structural invariant in each of the two blocks $H$ and $L$: the block's array is partitioned into a \emph{settled prefix} holding its $\mathit{lo}$ leftmost beads, an \emph{active segment} of length $\mathit{hi} - \mathit{lo}$ holding the beads of local position in $(\mathit{lo}, \mathit{hi}]$, and a \emph{settled suffix} holding the rest; no order is maintained within any part. In addition the two integers $r_H(\mathit{lo})$ and $r_L(\mathit{lo})$ are stored. Initially $\mathit{lo} = 0$, $\mathit{hi} = \beta$, and the active segments are the full blocks.

A round proceeds as follows with $\mathit{mid} = \bigl\lfloor \frac{\mathit{lo} + \mathit{hi}}{2} \bigr\rfloor$. In each block, run linear-time selection for local position $\mathit{mid} - \mathit{lo}$ within the active segment; this takes $\Theta(\mathit{hi} - \mathit{lo})$ time and splits the active segment into a left part (local positions $(\mathit{lo}, \mathit{mid}]$) and a right part (local positions $(\mathit{mid}, \mathit{hi}]$). Counting red beads in the two left parts yields $r_H(\mathit{mid}) = r_H(\mathit{lo}) + \#\set{\text{red in left part of } H}$ and analogously $r_L(\mathit{mid}) = r_L(\mathit{lo}) + \#\set{\text{red in left part of } L}$. If $\Delta(\mathit{mid}) \geq \tau$, set $\mathit{hi} \leftarrow \mathit{mid}$; the new active segments are the left parts, and the right parts are frozen into the settled suffixes. Otherwise set $\mathit{lo} \leftarrow \mathit{mid}$; the new active segments are the right parts, and the left parts are frozen into the settled prefixes. In addition, store the updated counts $r_H(\mathit{lo})$ and $r_L(\mathit{lo})$. In both cases, all invariants are maintained.

Rounds continue until $\mathit{hi} - \mathit{lo} = 1$, at which point the invariant $\Delta(\mathit{lo}) < \tau \leq \Delta(\mathit{hi})$ forces $\Delta(\mathit{hi}) = \tau$, and $d = \mathit{hi}$ is returned. The active segment halves each round, so the third stage takes $\Theta(\beta) + \Theta(\beta/2) + \Theta(\beta/4) + \dots = \Theta(\beta)$ time in total. The four boundary beads are available from the final selection calls: the beads of local positions $d$ and $d+1$ in $H$ sit at the last partition boundary of $H$, and likewise in $L$. These give us the two cut locations.
\end{proof}

\begin{theorem}
\label{thm:iteration}
One iteration of the unsorted one-dimensional subroutine sets apart a subset containing exactly $\frac{aR}{K}$ red and $\frac{aB}{K}$ blue beads with at most two cuts, and it runs in $\Theta(M \log K)$ time.
\end{theorem}

\begin{proof}
Lemma~\ref{lem:multiselect} partitions $P$ into blocks in $\Theta(M \log K)$ time. Lemma~\ref{lem:bracket} then either finds a block-aligned window with $f(i) = Q$ in $\Theta(K)$ time or brackets one between two consecutive block boundaries, in which case Lemma~\ref{lem:refine} locates it in $\Theta(\beta)$ time. Since $K \leq M$ and $\beta \leq M$, the total running time is $\Theta(M \log K)$. The resulting window contains exactly $Q = \frac{aR}{K}$ red and $s - Q = \frac{aB}{K}$ blue beads, and it is delimited by two endpoints, so it is set apart from the rest of the necklace with at most two cuts.
\end{proof}

\subsection{Main algorithm}
\label{sec:main_alg}
Now we can continue with our main algorithm. Decompose $k$ into a sum of powers of two: $k = 2^{e_1} + 2^{e_2} + \dots + 2^{e_\ell}$ with $e_1 > e_2 > \dots > e_\ell$. We treat each group of $2^{e_i}$ agents as a single \emph{superagent}, and first divide the points among the superagents. We project the input points onto the $x$-axis and perform $\ell-1$ iterations of the unsorted one-dimensional subroutine to separate the points into $\ell$ corresponding regions. We then recursively bisect each superagent's region. This results in $2(\ell-1) + \sum_{i=1}^\ell (2^{e_i}-1) = 2(\ell-1) + k - \ell = k + \ell - 2$ total cuts. The value $\ell$ is equivalent to the Hamming weight of the binary representation of $k$. The smallest integer with Hamming weight $\ell$ is $2^\ell - 1$, so $k \geq 2^\ell - 1$, hence $\ell \leq \log_2(k+1)$. Thus, our algorithm generates at most $k + \log_2(k+1) - 2$ cuts for all $k$.

If $\ell = 1$, we do not need to run the unsorted one-dimensional subroutine. If $\ell > 1$, assuming we do the one-dimensional cuts in descending order of allocation size, the subroutine contributes $\Theta\Bigl(\sum_{j=1}^{\ell-1} \Bigl(\frac{m}{k} \sum_{i=j}^\ell 2^{e_i}\Bigr) \log\Bigl(\sum_{i=j}^\ell 2^{e_i}\Bigr)\Bigr) = \Theta(m \log k)$ time (calculations in Section~\ref{sec:calculations}). Each ham sandwich cut takes linear time~\cite{Lo1994}, so together they contribute $\Theta(\frac{m}{k} \sum_{i=1}^\ell e_i 2^{e_i}) = \Theta(m \log k)$ time (calculations in Section~\ref{sec:calculations}). Thus, the overall time complexity is $\Theta(m \log k)$.
A sample execution is shown in Figure~\ref{fig:general_alg}.

\begin{theorem}
\label{thm:general_alg}
For every input, a solution using at most $k + \log_2(k+1) - 2$ cuts can be computed in $\Theta(m \log k)$ time.
\end{theorem}

\begin{figure}[bhpt]
\centering
\includegraphics[page=7]{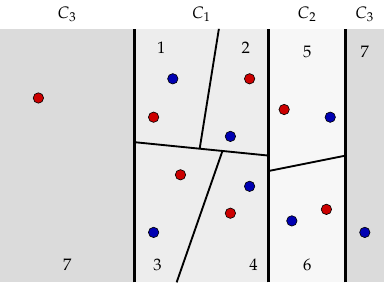}
\caption{A visualization of the general algorithm with $k = 7 = 2^2 + 2^1 + 2^0$ and $m=14$. The two iterations of the unsorted one-dimensional subroutine produce the thick vertical cuts, separating superagents $C_1$, $C_2$, and $C_3$ (the two outer pieces). Each superagent is then recursively bisected by ham sandwich cuts. Cells are numbered by agent; Agent $7$ receives two cells.}
\label{fig:general_alg}
\end{figure}

\subsection{Convex position}
\label{sec:convex}
If the input is in strictly convex position, we can show that $k-1$ cuts suffice. A point set is in \emph{strictly convex position} if every point is a vertex of the convex hull and no three points are collinear.

\begin{prop}
\label{prop:convex}
For input in strictly convex position, a solution using at most $k-1$ cuts can be computed in $\Theta(m \log k)$ time.
\end{prop}

\begin{proof}
Select an arbitrary viewing point in the convex hull of the input (e.g., the centroid of three arbitrary input points, which can be found in constant time). Consider the angle of each input point with respect to the viewing point; if sorted, these would correspond exactly to the circular ordering of their positions on the boundary of the hull. Using the angles, the points can thus be treated as a circular necklace.

We then run the algorithm of Section~\ref{sec:main_alg} with two modifications: we use the input point's angles instead of their $x$-coordinates for the unsorted one-dimensional subroutine; and during each iteration of the subroutine, we produce the two desired cut points using a single generalized cut (by convexity, we can carve out the interval with a single line between the endpoints of the interval). In total, this procedure makes $\ell-1 + \sum_{i=1}^\ell (2^{e_i}-1) = k-1$ cuts and takes $\Theta(m \log k)$ time.
\end{proof}

In this section, we used the convexity of the input to save $\ell-1$ cuts, but this does not rule out the possibility of reducing the number of cuts for arbitrary input by some other technique.

\begin{conj}
For every input, $k-1$ cuts suffice.
\end{conj}

\section{Dynamic Algorithm}
\label{sec:dynamic}
Here we adapt our problem formulation to a dynamic setting: given a solution produced by the static algorithm of Section~\ref{sec:main_alg}, we maintain a fair set of cuts under \emph{relocation} of input points to and from arbitrary positions. The fairness constraint and the upper bound on the number of cuts are maintained deterministically for arbitrary updates, and the worst-case update time is that of one static run. We then analyze the expected update time under two bounded-displacement models: a stochastic model, in which the guarantee holds for every update of an arbitrarily long sequence, and an adversarial model, in which it holds for a single batch applied to a random instance.

\subsection{Overview}
A single relocation works using one of two mechanisms, depending on whether or not the old and new positions lie in the same superagent's region of the top-level $x$-decomposition. If they lie in the same superagent's region, a ham sandwich cut can become unfair only if the point counts on its sides change, so it suffices to rebuild only the subtree rooted at the deepest node whose cell contains both positions; a relocation within a single leaf cell costs nothing. If they lie in regions belonging to distinct superagents, we rebalance by rerunning the static algorithm on a connected superset of superagents, chosen to minimize total rebuild cost.

Both mechanisms are special cases of a general batch algorithm, developed below. In every case the produced set of cuts is one the static algorithm could have produced, so exact fairness is restored and the cut count stays at most $k+\ell-2 \leq k+\log_2(k+1)-2$; the worst-case update time is $O(m\log k)$, attained when the repair spans a constant fraction of the input.

\subsection{Model}
\label{sec:model}
Throughout, the static structure is as in Section~\ref{sec:main_alg}: superagents $C_1,\dots,C_\ell$ with $|C_i| = 2^{e_i} \frac{m}{k}$, separated by the $2(\ell-1)$ top-level cuts of the one-dimensional subroutine, each $C_i$ recursively divided by ham sandwich cuts, resulting in a bisection tree of depth $e_i$ where each node $v$ has a cell (the intersection of $C_i$'s region with the halfplanes on the root-to-$v$ path) containing $n_v = \frac{|C_i|}{2^{d_v}}$ points at depth $d_v$. Each point stores its superagent and leaf; each node stores its cut line and per-color counts.

We also introduce the \emph{neighborhood graph} $G$, an undirected, vertex-weighted graph analogous to the one used by Advani et al.~\cite{Advani2026a}. The vertices are the superagents $C_1, \dots, C_\ell$, and there is an edge between two vertices if the corresponding superagents' regions contain adjacent intervals (i.e., border the same top-level cut). Vertex $C_i$ has weight $\lvert C_i\rvert e_i$, the cost of rebuilding its bisection tree from scratch. Note that $G$ is connected and has at most $2(\ell-1)$ edges.

A batch consists of $p$ relocations $(x_1 \to x_1'), \dots, (x_p \to x_p')$ of distinct input points. A node $v$ in a bisection tree is \emph{valid} if each child of $v$ contains exactly its quota of each color. For a batch, the \emph{demand} of superagent $C_i$ is $\Delta(C_i) = (\Delta r_i, \Delta b_i) \in \Integers^2$, the net change in its red and blue counts. A set $S$ of superagents is \emph{repairable} if the repair that reruns the one-dimensional subroutine on each connected component in $G[S]$ (a \emph{cluster}) restores exact fairness. $G[S]$ denotes the subgraph of the neighborhood graph $G$ induced by the vertex set $S$.

\subsection{Batch relocation algorithm}
\label{sec:batch}
We now describe the steps of the batch relocation algorithm. A sample execution is shown in Figure~\ref{fig:dynamic}.
\begin{enumerate}
\item \textbf{Move and tally.} For each point to be relocated, the source superagent is stored with the point; the destination superagent is found by binary search over the $x$-coordinates of the top-level cuts, followed by an interval-to-superagent lookup, in $O(\log \ell)$ time. Decrement the per-color counts of every node on the source's old root-to-leaf path, and increment them on the destination's root-to-leaf path, in $O(\log k)$ time per relocated point, and record the demand vectors $\Delta(C_i) = (\Delta r_i, \Delta b_i)$.
\item \textbf{Cluster selection.} Let $N = \set{C_i \suchthat \Delta(C_i) \neq (0,0)}$.
\begin{itemize}
\item If $N = \emptyset$: take $S^* = \emptyset$ (no repair needed).
\item If $N = \set{C_a, C_b}$: take $S^*$ to be the vertex set of a minimum-weight path from $C_a$ to $C_b$ in $G$, computed in $O(\log k \log\log k)$ time by Dijkstra's algorithm~\cite{Fredman1984}, or $O(\log k)$ time via Thorup's algorithm~\cite{Thorup1999}, after the standard reduction of vertex weights to edge weights (Lemma~\ref{lem:twoterm}).
\item Otherwise ($|N| > 2$): compute $S^*$ as the cheaper of the two algorithms of Section~\ref{sec:selection} (Theorems~\ref{thm:cluster_1}~and~\ref{thm:cluster_2}).
\end{itemize}
\item \textbf{Local repair of balanced superagents.} For each superagent with $\Delta(C_i) = (0,0)$ that does not belong to $S^*$, scan the processed nodes top-down to find the maximal invalid nodes, and rebuild the subtree of each such node $u$ by recursive ham-sandwich bisection of the points now in its cell, at cost $O(n_u(e_i - d_u))$.
\item \textbf{Cluster rebuild.} For each cluster in $G[S^*]$ with superagent index set $I$: remove all cuts interior to the cluster's region (the top-level cuts between its superagents and all bisection cuts within them), pool its points, re-separate them into $|I|$ superagents via $|I|-1$ iterations of the unsorted one-dimensional subroutine, and rebuild each superagent's bisection tree. Finally, update the stored superagent and leaf pointers of the affected points and the per-color counts.
\end{enumerate}

\begin{figure}[bhpt]
\centering
\begin{subfigure}[b]{0.45\linewidth}
\centering
\begin{minipage}[t][7.7cm][s]{\linewidth}
\centering
\includegraphics[page=8]{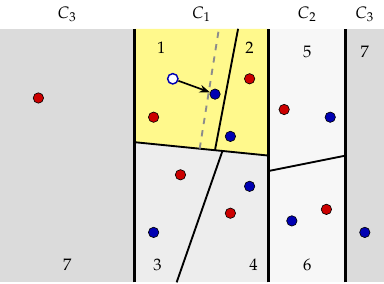}

\vfill

\includegraphics[page=9]{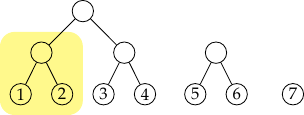}
\end{minipage}
\caption{Within-superagent relocation.}
\label{fig:dynamic_a}
\end{subfigure}
\hfill
\begin{subfigure}[b]{0.45\linewidth}
\centering
\begin{minipage}[t][7.7cm][s]{\linewidth}
\centering
\includegraphics[page=10]{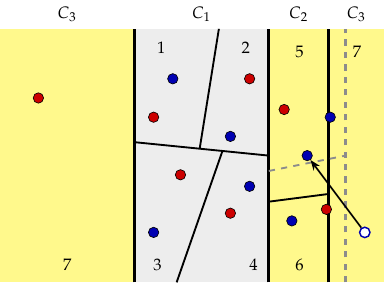}

\vfill

\includegraphics[page=11]{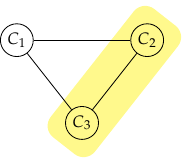}
\end{minipage}
\caption{Cross-superagent relocation.}
\label{fig:dynamic_b}
\end{subfigure}
\caption{Two single relocations on the instance of Figure~\ref{fig:general_alg}. Hollow points mark old positions, and dashed cuts are discarded. Yellow shading indicates cells/superagents affected by the update.}
\label{fig:dynamic}
\end{figure}

\subsection{Deterministic guarantees}
\begin{lemma}[Correctness of local repair]
\label{lem:welldef}
Let $u$ be a maximal invalid node in a superagent with zero net demand. Then its cell contains exactly its quota of each color, so the rebuild in Step~3 restores exact fairness within the subtree of $u$.
\end{lemma}

\begin{proof}
If $u$ is the root of the bisection tree, its cell is the full superagent, whose totals equal its quotas because the superagent's net demand is zero. Otherwise, by maximality the parent of $u$ is valid, and validity of a node means precisely that the cell of each of its children --- in particular that of $u$ --- holds exact quota counts of both colors. Recursively bisecting a set with exact totals yields exact quotas all the way to the leaves.
\end{proof}

\begin{lemma}[Repairability]
\label{lem:feasible}
A set $S$ of superagents is repairable if and only if every cluster in $G[S]$ has zero aggregate demand in each color.
\end{lemma}

\begin{proof}
($\impliedby$) A cluster with superagents $\set{C_i}_{i\in I}$ holds exactly $\sum_{i\in I} 2^{e_i} \frac{r}{k}$ red and $\sum_{i\in I} 2^{e_i} \frac{b}{k}$ blue points. These totals are divisible as required, so $|I|-1$ iterations of the one-dimensional subroutine separate the cluster into superagents with exact quotas, and the ham sandwich cuts finish the repair.

($\implies$) The repair modifies only $S$; superagents outside $S$ retain their exact point sets. Hence the point set available to each cluster is fixed, and if its aggregate demand is nonzero in some color, no allocation within the cluster can meet the quotas.
\end{proof}

\begin{lemma}[Two-terminal repair]
\label{lem:twoterm}
Suppose exactly two superagents $C_a$ and $C_b$ have nonzero demand. Then a set $S$ is repairable if and only if $C_a$ and $C_b$ lie in the same cluster in $G[S]$, and the minimum-weight such set $S$ is the vertex set of a shortest path between $C_a$ and $C_b$ in $G$.
\end{lemma}

\begin{proof}
By Lemma~\ref{lem:feasible}, $S$ is repairable if and only if every cluster in $G[S]$ has zero aggregate demand. Every vertex outside $\set{C_a, C_b}$ has zero demand on its own, so a cluster is automatically balanced unless it contains exactly one of $C_a$ and $C_b$; hence $S$ is repairable if and only if $C_a$ and $C_b$ are in the same cluster. Among sets $S$ with $C_a, C_b$ in the same cluster, the minimum-weight choice is achieved by a set inducing a single path between them: any connected vertex set containing both contains the vertex set of some path between them, and adding vertices beyond a chosen path never decreases weight, since all weights are nonnegative. Minimizing over paths gives a vertex-weighted shortest path, computable by the standard reduction of vertex weights to edge weights followed by Dijkstra's algorithm in $O(\ell \log \ell) = O(\log k \log\log k)$ time, or by Thorup's algorithm in $O(\ell) = O(\log k)$ time.
\end{proof}

\begin{prop}[Correctness]
\label{prop:batch_correct}
The batch relocation algorithm restores exact fairness, maintains at most $k + \ell - 2 \leq k + \log_2(k+1) - 2$ cuts, and runs in $\Theta(m\log k)$ time.
\end{prop}

\begin{proof}
Fairness: for superagents with net zero demand, Step~3 restores all quotas by Lemma~\ref{lem:welldef}; clusters are repaired in Step~4 by Lemma~\ref{lem:feasible} (or, when $|N|\leq 2$, by Lemma~\ref{lem:twoterm}).

Cut count: the rebuilt configuration is one that could have been produced by an execution of the static algorithm, mirroring the invariant argument of Advani et al.~\cite{Advani2026a}; the static bound then applies.

Time: updating per-color counts (Step~1) costs $O(\log k)$ per relocation; cluster selection (Step~2) costs $O(2^\ell \ell) = O(k\log k)$ in the worst case (Theorem~\ref{thm:cluster_1}), or less (Section~\ref{sec:selection}); and repair (Step~3) plus rebuilding (Step~4) is at most one static run, $\Theta(m\log k)$.
\end{proof}

\subsection{Repair-region selection}
\label{sec:selection}
Advani et al.~\cite{Advani2026a} formalize the choice of which agents to rebuild as the \textsf{MinNodeMaxFlow} problem and prove it NP-complete, which motivates their spanning-tree approximation. In our setting the neighborhood graph has only $\ell \leq \log_2(k+1)$ vertices, so the optimal repair region can be computed exactly in polynomial time.

\begin{theorem}
\label{thm:cluster_1}
A repairable superset $S^*$ of the nonzero-demand superagents minimizing the total rebuild weight $\sum_{C_i \in S} \lvert C_i \rvert e_i$ can be computed in $O\bigl(2^\ell \ell\bigr) = O(k \log k)$ time.
\end{theorem}

\begin{proof}
There are at most $2^\ell \leq k+1$ candidate sets $S$. For each, compute the clusters in $G[S]$ and their aggregate demands; since $G$ has at most $2(\ell - 1)$ edges, each check costs $O(\ell)$ given the precomputed demand vectors. Retain the minimum-weight set with zero aggregate demand (by Lemma~\ref{lem:feasible}, it is repairable).
\end{proof}

Theorem~\ref{thm:cluster_1} is tight in the worst case, but a batch of $p \ll k$ relocations touches very few superagents: each single relocation changes the demand of at most two superagents, so $|N| \leq 2p$ for $N = \set{C_i \suchthat \Delta(C_i) \neq (0,0)}$. This can be exploited directly.

\begin{theorem}
\label{thm:cluster_2}
Let $t = |N|$. A repairable superset $S^*$ of the nonzero-demand superagents minimizing the total rebuild weight $\sum_{C_i \in S} \lvert C_i \rvert e_i$ can be computed in $O\bigl(3^t \ell + 2^t \ell \log \ell\bigr)$ time.
\end{theorem}

\begin{proof}
By Lemma~\ref{lem:feasible}, a set $S$ is repairable if and only if every cluster in $G[S]$ has zero aggregate demand. Fix an optimal $S^*$. The clusters in $G[S^*]$ induce a partition $P$ of $N$ into groups, one per cluster, each with zero aggregate demand. Each cluster is a connected vertex set containing its group, so its weight is at least the minimum-weight (vertex-weighted) Steiner tree connecting that group in $G$. Hence
\[
\operatorname{weight}(S^*) \geq \sum_{\mathcal{T} \in P} \operatorname{steiner}(\mathcal{T}) \geq \min_{P'} \sum_{\mathcal{T} \in P'} \operatorname{steiner}(\mathcal{T}),
\]
where the minimum is over all partitions $P'$ of $N$ into zero-demand groups.

Conversely, fix a minimizing partition $P^*$ and let $S$ be the union, over $\mathcal{T} \in P^*$, of a minimum Steiner tree connecting each $\mathcal{T}$. Any cluster in $G[S]$ is a union of one or more complete Steiner trees from this collection (trees never split across clusters, since there are no edges between clusters), so its aggregate demand is a sum of zero-demand groups and hence zero; $S$ is therefore repairable. Its weight is at most $\sum_{\mathcal{T}\in P^*}\operatorname{steiner}(\mathcal{T})$, since shared vertices between trees are counted once in $S$ but possibly multiple times in the sum. Combining both directions, the true optimum equals $\min_{P} \sum_{\mathcal{T} \in P} \operatorname{steiner}(\mathcal{T})$.

Now we show how to solve the Steiner tree problem on a cluster. We reduce to a directed, edge-weighted instance: replace each edge $(u,v)$ by directed edges $(u,v)$ and $(v,u)$ with costs $w(v)$ and $w(u)$ respectively. Since every non-root vertex of an arborescence has exactly one incoming edge, the weight of an arborescence rooted at $r$ equals $\sum_{x \neq r} w(x)$ over its vertices $x$, i.e.\ the original objective minus $w(r)$. For $v \in V(G)$ and $X \subseteq N$, let $A(v, X)$ denote the minimum weight of an arborescence rooted at $v$ containing $X$, with $A(v, \emptyset) = 0$; then $\operatorname{steiner}(\mathcal{T}') = w(r) + A(r, \mathcal{T}' \setminus \set{r})$ for any $r \in \mathcal{T}'$. Using the Dreyfus--Wagner dynamic program~\cite{Dreyfus1971,Erickson1987}, we compute $A(v, X)$ for every $X \subseteq N$, yielding $\operatorname{steiner}(\mathcal{T}')$ for every $\mathcal{T}' \subseteq N$, in $O\bigl(3^t \ell + 2^t \ell \log \ell\bigr)$ time on our graph with $\ell$ vertices and $O(\ell)$ edges.

It remains to assemble an optimal balanced partition from these values. Define $f(\mathcal{T}) = \min_{\emptyset \neq \mathcal{T}' \subseteq \mathcal{T},\ \Delta(\mathcal{T}') = (0,0)} \operatorname{steiner}(\mathcal{T}') + f(\mathcal{T} \setminus \mathcal{T}')$ with $f(\emptyset) = 0$; then $f(N)$ equals the partition optimum established above. Evaluating $f$ by submask enumeration takes $O(3^t)$ time, with the aggregate demands $\Delta(\mathcal{T}')$ of all subsets precomputed in $O(2^t)$ time, and the vertex sets of the trees in an optimal partition, recovered by standard traceback, give $S^*$.
\end{proof}

Combining Theorems~\ref{thm:cluster_1} and~\ref{thm:cluster_2}, cluster selection costs $O\bigl(\min\set[\big]{2^\ell \ell, 3^t\ell + 2^t \ell \log\ell}\bigr)$ time, where $t = |N| \leq \min(\ell, 2p)$; the batch relocation algorithm can compare the two costs in advance --- both depend only on $N$ and $G$, not on the rebuild itself --- and run whichever is cheaper. Comparing the leading terms, the crossover between the two algorithms occurs at $t = \ell \log_3 2 \approx 0.63\ell$: the enumeration of Theorem~\ref{thm:cluster_1} is preferable for larger $t$, and the parameterized algorithm of Theorem~\ref{thm:cluster_2} for smaller $t$. In particular, since $t \leq 2p$, any batch with $p < \frac{1}{2}\ell\log_3 2$ falls in the latter regime.

\subsection{Distance-sensitive expected update time}
A relocation invalidates only the subtree of the deepest node whose cell contains both endpoints. The closer the start and end positions are, the smaller that subtree is likely to be. Using that idea, we now analyze the expected running time of the batch relocation algorithm.

We consider two update models. In both models, the $m$ initial point positions are drawn i.i.d.\ uniformly from $[0,1]^2$, and the coloring is arbitrary, even adversarial; a batch update chooses $p$ distinct points uniformly at random. Set $D = \sum_i \delta_i$ and $\varepsilon_m = c_0 \sqrt{(\log m)/m}$, where $c_0$ is a sufficiently large absolute constant fixed in the proof of Lemma~\ref{lem:strip}.
\begin{description}
\item[Adversarial displacement] The $i$th chosen point is moved by an adversarial displacement vector of length at most $\delta_i$, chosen after the point is revealed, with the resulting segment clipped to $[0,1]^2$. The guarantees below hold for the first batch applied to a fresh instance; after an adversarial batch the positions are no longer i.i.d., and we make no claim for subsequent batches.
\item[Stochastic displacement] The $i$th chosen point is moved by a random vector $V_i$ drawn independently from a fixed absolutely continuous distribution $\mu_i$ supported on the disk of radius $\delta_i \le 1$ and invariant under reflection in either coordinate axis (e.g., the uniform distribution); the moved point is \emph{reflected} back into the square by the folding map $\pi \colon [-1,2]^2 \to [0,1]^2$ given by $\pi(z) = (\varphi(z_1), \varphi(z_2))$, where $\varphi(t) = \min(\lvert t \rvert,\, 2 - t)$ reflects across $0$ and $1$. The guarantees below hold for every batch of an arbitrarily long sequence.
\end{description}

\begin{lemma}[Stationarity]
\label{lem:stationary}
In the stochastic model, the point positions are i.i.d.\ uniform on $[0,1]^2$ before every update. Moreover, the folded path of a moved point has length at most $\delta_i$, so a point that crosses a cut $\lambda$ starts within distance $\delta_i$ of $\lambda \cap [0,1]^2$.
\end{lemma}

\begin{proof}
Since $\lVert V_i \rVert \leq \delta_i \leq 1$, a moved point starting at $x \in [0,1]^2$ satisfies $x + V_i \in [-1,2]^2$, the domain of $\pi$. For $s \in (0,1)$, the preimages of $s$ under $\varphi$ in $[-1,2]$ are exactly $s$, $-s$, and $2 - s$; that is, $\tau(s)$ for $\tau$ ranging over $\mathcal{T} = \set{\operatorname{id}, t \mapsto -t, t \mapsto 2-t}$. Hence for $y \in (0,1)^2$, the preimages of $y$ under $\pi$ in $[-1,2]^2$ are the nine points $\tau(y) \coloneqq (\tau_1(y_1), \tau_2(y_2))$ with $\tau = (\tau_1, \tau_2) \in \mathcal{T}^2$.

Let $f$ be the density of $\mu_i$. Each $\tau \in \mathcal{T}^2$ is an isometry of $\Reals^2$, and the nine sets $\tau((0,1)^2)$, $\tau \in \mathcal{T}^2$, are pairwise disjoint, so the new position $\pi(x + V_i)$ has density at $y \in (0,1)^2$ given by
\[
\kappa(x, y) = \sum_{\tau \in \mathcal{T}^2} f(\tau(y) - x).
\]
We claim each summand is symmetric in $x$ and $y$. In a coordinate $j$ with $\tau_j(t) = -t$ we have $\tau_j(y_j) - x_j = -y_j - x_j = \tau_j(x_j) - y_j$, and with $\tau_j(t) = 2 - t$ likewise $\tau_j(y_j) - x_j = 2 - y_j - x_j = \tau_j(x_j) - y_j$; in a coordinate with $\tau_j = \mathrm{id}$, the quantities $y_j - x_j$ and $x_j - y_j$ differ by a sign, and $f$ is invariant under flipping the sign of either coordinate. Hence $f(\tau(y) - x) = f(\tau(x) - y)$ for every $\tau$, and summing gives $\kappa(x, y) = \kappa(y, x)$.

Starting from any $z \in [0,1]^2$, the moved point lands in $[0,1]^2$ with probability $1$, so $\int_{[0,1]^2} \kappa(z, x)\, dx = 1$. If a point's position is uniform on $[0,1]^2$, the density of its updated position at $y$ is therefore
\[
\int_{[0,1]^2} \kappa(x, y)\, dx = \int_{[0,1]^2} \kappa(y, x)\, dx = 1,
\]
so the updated position is again uniform on $[0,1]^2$. Since the initial positions are i.i.d.\ uniform on $[0,1]^2$, the first claim follows by induction on updates.

Finally, the segment from $x$ to $x + V_i$ lies in $[-1,2]^2$ and is monotone in each coordinate. In each coordinate it sweeps an interval of length at most $\delta_i \leq 1$ that meets $[0,1]$, so it crosses a fold line at most once in that coordinate; on each piece between crossings, $\pi$ acts as an isometry. Hence $\pi$ maps the segment to a polygonal path in $[0,1]^2$ with at most two bends, of the same total length $\lVert V_i \rVert \leq \delta_i$, starting at $\pi(x) = x$. If this path crosses a cut $\lambda$, the crossing point lies in $\lambda \cap [0,1]^2$ and its distance from $x$ is at most the path length, so the point started within distance $\delta_i$ of $\lambda \cap [0,1]^2$.
\end{proof}

\begin{lemma}[Strip bound]
\label{lem:strip}
With probability at least $1 - \frac{1}{m^2}$ over the random draw of the $m$ input points, the following holds simultaneously for every line $\lambda$ and every $\delta > 0$: the fraction of input points within distance $\delta$ of $\lambda \cap [0,1]^2$ is at most $2\delta \bigl(\sqrt{2} + 2\delta\bigr) + \varepsilon_m$.
\end{lemma}

\begin{proof}
The $\delta$-neighborhood of $\lambda \cap [0,1]^2$ is contained in a rectangle of side lengths $\sqrt{2} + 2\delta$ and $2\delta$. Since the neighborhood is contained in the rectangle, the fraction of input points in the neighborhood is at most the fraction in the rectangle, so a bound on the latter suffices. Rectangles of arbitrary orientation are intersections of four halfplanes, a range family of constant VC dimension, so by the $\varepsilon$-sample theorem~\cite{Har-Peled2011}, with probability at least $1 - \frac{1}{m^2}$, the empirical measure of every such rectangle simultaneously deviates from its area by at most $\varepsilon_m$, provided the constant $c_0$ in the definition of $\varepsilon_m$ is chosen sufficiently large relative to the VC dimension of the family and the absolute constant of the $\varepsilon$-sample theorem. The area of the rectangle is $2\delta \bigl(\sqrt{2} + 2\delta\bigr)$, which gives the claim.
\end{proof}

\begin{theorem}[Single relocation]
\label{thm:single}
Consider a single update with displacement bound $\delta$, either in the adversarial model applied to a fresh instance, or in the stochastic model at any point of an update sequence. Then
\[
\E[T] = O\bigl((\delta + \varepsilon_m) m (\log k)^2 + \log k\bigr).
\]
In particular $\E[T] = O(\delta m(\log k)^2 + \log k)$ whenever $\delta \geq \sqrt{(\log m)/m}$. The bound holds for arbitrary colorings; in the adversarial model, the displacement vector may be chosen adaptively.
\end{theorem}

\begin{proof}
In the stochastic model, the expectation is over the current configuration, which is i.i.d.\ uniform by Lemma~\ref{lem:stationary}; in the adversarial model it is i.i.d.\ uniform by assumption. In either case the crossing argument below applies to the (possibly folded) path of the moved point, which by Lemma~\ref{lem:stationary} lies within distance $\delta$ of its starting point.

Updating per-color counts along two root-to-leaf paths costs $O(\log k)$; if no cut is invalidated, the update ends there. Otherwise, let $x$ and $x'$ be the old and new positions, with $|x - x'| \leq \delta$ and the segment $xx' \subseteq [0,1]^2$. A single relocation leaves at most two superagents with nonzero demand, so by Lemma~\ref{lem:twoterm}, the cluster selection of Step~2 costs $O(\log k \log\log k)$ time via Dijkstra's algorithm, or $O(\log k)$ via Thorup's algorithm, dominated by the rebuild cost below.

First consider the case in which $x$ and $x'$ lie in the region of the same superagent $C_i$. The invalid cuts are exactly those at nodes on the two root-to-leaf paths strictly below their deepest common node $u$, together with the cut at $u$; the repair rebuilds the subtree of $u$ at cost $O(n_u(e_i - d_u))$. Invalidation of the cut $\lambda_v$ at any node $v$ requires $x$ and $x'$ to lie on opposite sides of the cut; since $xx' \subseteq [0,1]^2$, the segment crosses $\lambda_v \cap [0,1]^2$, so $x$ lies within $\delta$ of that set. The rebuild cost is $O(n_u(e_i - d_u))$ for the maximal invalid node $u$, which is at most $\sum_v \mathbf{1}[\lambda_v \text{ crossed}] \cdot O(n_v(e_1 - d_v))$, where the sum ranges over all nodes $v$; summing over all nodes accounts for whichever cuts the adversarially chosen displacement actually crosses, so no case analysis over the adversary's choice is needed. Since the source point is uniform over the $m$ input points, and since Lemma~\ref{lem:strip} holds simultaneously for all lines --- in particular for the data-dependent cuts $\lambda_v$ --- each crossing indicator has expectation at most $2\delta \bigl(\sqrt{2} + 2\delta\bigr) + \varepsilon_m$ on the success event of the lemma. As $\delta \leq \sqrt{2}$, this is $O(\delta + \varepsilon_m)$, a bound independent of $v$. Hence
\begin{align*}
\E[T_{\mathrm{rebuild}}]
&\leq \sum_v O(\delta + \varepsilon_m)
\cdot O\bigl(n_v (e_1 - d_v)\bigr) \\
&= O(\delta + \varepsilon_m) \cdot \sum_{v} O\bigl(n_v (e_1 - d_v)\bigr) \\
&= O\bigl((\delta + \varepsilon_m) m(\log k)^2\bigr) \,,
\end{align*}
where the last equality holds by grouping nodes by depth: at any fixed depth $j$, the cells are pairwise disjoint, so $\sum_{v : d_v = j} n_v \leq m$, and therefore
\begin{align*}
\sum_v n_v(e_1 - d_v)
&= \sum_{j=0}^{e_1} (e_1 - j) \sum_{v : d_v = j} n_v \\
&\leq m \sum_{j=0}^{e_1} (e_1 - j) \\
&= m \cdot \frac{e_1(e_1+1)}{2} \\
&= O\bigl(m (\log k)^2\bigr) \,,
\end{align*}
using $e_1 \leq \log_2 k$. If instead $xx'$ crosses one of the (at most) $2(\ell - 1)$ top-level cuts, the cross-superagent repair of Section~\ref{sec:batch} costs at most $O(m \log k)$, and the same strip bound applies to each top-level cut, contributing $O((\delta + \varepsilon_m) \ell m \log k) = O\bigl((\delta + \varepsilon_m) m(\log k)^2\bigr)$. An adaptive adversary chooses which nearby cuts to cross, but the displayed sum ranges over all cuts, so it nullifies the power of any such choice. On the failure event of Lemma~\ref{lem:strip} (probability at most $m^{-2}$), the cost is at most that of a static run, $O(m \log k)$, contributing $O\bigl(\frac{\log k}{m}\bigr) = o(1)$ to $\E[T]$.
\end{proof}

\begin{corollary}[Sub-cell displacements are nearly free]
With probability at least $1 - O(k(\delta + \varepsilon_m))$, a single update invalidates no cut and completes in $O(\log k)$ time.
\end{corollary}

\begin{proof}
We use a union bound over the at most $k+\ell-2$ cuts, applying Lemma~\ref{lem:strip} to each. When no cut is invalidated, $N=\emptyset$ and Step~2 of the relocation algorithm costs $O(1)$, so the update completes in $O(\log k)$ time (due to Step~1).
\end{proof}

\begin{theorem}[Batch relocation]
\label{thm:batch}
For a batch with displacement bounds $\delta_1,\dots,\delta_p$ and $D = \sum_i \delta_i$, under the same conditions as Theorem~\ref{thm:single},
\[
\E[T_{\mathrm{batch}}]
= O\bigl(\min\set{(D + p\varepsilon_m)\log k, 1}
\cdot m\log k + (p + k)\log k\bigr).
\]
\end{theorem}

\begin{proof}
Overhead: $O(\log k)$ per single relocation for updating per-color counts, plus $O\bigl(2^\ell\ell\bigr) = O(k\log k)$ for cluster selection in the worst case (Theorem~\ref{thm:cluster_1}). For rebuilding, we bound the batch cost by both the sum of single-move costs and the cost of rerunning the static algorithm, and then take the minimum; the expectation of the minimum is at most the minimum of the two bounds.

For the sum of single-move costs, first observe that a cut can be invalid after the batch only if the final point counts on its two sides differ from their quotas, which requires at least one individual move to have crossed it; hence every cut invalid after the batch is also invalidated by some individual move, considered in isolation. Now let $r$ be any maximal invalid node of the batch. By the observation, the cut at $r$ is crossed by some move $i$, so $r$ lies inside the invalid subtree of move $i$, whose own maximal invalid node $r_i$ is therefore an ancestor of $r$ (or $r$ itself). We charge the rebuild cost of $r$ to move $i$. A single move $i$ may receive charges from several batch-maximal nodes, but any two distinct batch-maximal nodes are incomparable in the tree, so the charged nodes below $r_i$ head pairwise disjoint subtrees within the subtree of $r_i$; for disjoint subtrees inside the subtree of $r_i$, the rebuild costs satisfy $\sum_r n_r(e_1 - d_r) \leq n_{r_i}(e_1 - d_{r_i})$, i.e., the total charge to move $i$ is at most the cost that move $i$ would have incurred alone. Summing over moves, the total within-superagent rebuild cost of the batch is at most the sum of the individual bounds from Theorem~\ref{thm:single}. For cross-superagent repair, let $S_i$ be the set of superagents whose regions meet the closed $x$-interval spanned by move $i$'s endpoints; each $G[S_i]$ is connected. Let $U$ be the vertex set of any connected component of $G[\bigcup_i S_i]$. If a move $j$ has an superagent in $U$, then $S_j \cap U \neq \emptyset$, and since $G[S_j]$ is connected and $U$ is a connected component of $G[\bigcup_i S_i]$, we have $S_j \subseteq U$; in particular both of move $j$'s superagents lie in $U$. The set $U$ thus has zero aggregate demand, so $\bigcup_i S_i$ is repairable by Lemma~\ref{lem:feasible}. An optimal repairable superset therefore costs at most $\sum_i \operatorname{weight}(S_i)$, and each term is bounded in expectation exactly as the cross-superagent term in the proof of Theorem~\ref{thm:single}. Linearity of expectation gives the sum $\sum_i O\bigl((\delta_i + \varepsilon_m)m(\log k)^2\bigr) = O\bigl((D + p\varepsilon_m) m(\log k)^2\bigr)$.
\end{proof}

\begin{corollary}[Small batches]
For any $p$, the overhead term of Theorem~\ref{thm:batch} improves to
\[
O\bigl(p\log k + \min\set[\big]{k\log k,\ 9^p\log k + 4^p \log k\log\log k}\bigr).
\]
In particular, if $p = O(1)$, cluster selection costs $O(\log k \log\log k)$ rather than $O(k\log k)$, and the overhead term is $O(p\log k + \log k \log\log k)$.
\end{corollary}

\begin{proof}
The claim follows from Theorem~\ref{thm:cluster_2} with $\ell = O(\log k)$ and $t = |N| \leq 2p$.
\end{proof}

\subsection{Other dynamic updates}
The dynamic algorithm can be easily adapted to support various other updates.
\begin{description}
\item[Recoloring] Any recoloring of a subset of points that leaves the total number of points of each color unchanged can be simulated by relocations. Match each point changing from red to blue with a distinct point changing from blue to red, and swap the positions of each matched pair; the resulting point set is identical to the recolored one.
\item[Insertion] Any insertion of $k$ points of the same color can be simulated by relocations. Add one point of the desired color to each agent's region. Relocate the $k$ new points to the desired positions.
\item[Deletion] Any deletion of $k$ points of the same color can be simulated by relocations. Temporarily remove one point of the desired color from each agent's region. Relocate the $k$ points to be deleted to the newly vacant positions.
\end{description}
All results proved for relocation hold in these adapted settings.

\section{Conclusion}
\label{sec:conc}
We initiated the algorithmic study of necklace splitting in the plane with straight-line cuts, giving static algorithms that provably beat the $2(k-1)$ projection bound --- down to $k + \log_2(k+1) - 2$ cuts in $\Theta(m \log k)$ time for arbitrary $k$ and arbitrary input, and further still under structural assumptions --- together with a dynamic algorithm that maintains exact fairness under relocations with distance-sensitive expected update
time.

The central open question is the gap between our $k-1$ lower bound and $k + \log_2(k+1) - 2$ upper bound; we conjecture that $k-1$ cuts suffice in the worst case. Several more specific directions also remain. We conjecture that two cuts suffice for $k = 3$. We also conjecture that the fraction of evenly spaced instances with $k = 3$ requiring the full $2(k-1)$ cuts vanishes as $m$ tends to infinity, and it is an open question whether or not the cases $m = 5k$, $m = 7k$, and larger multiples require $2(k-1)$ cuts. Next, the decision version of the problem is clearly in NP; is the problem NP-complete? Finally, natural extensions include input with more than two colors and higher dimensions.

\bibliographystyle{plainnat}  
\bibliography{main}

\clearpage
\appendix

\section{Jordan-Curve Cuts with Many Colors}
\label{sec:jordan_curve}
In this section, we consider a variant of our problem where cuts, instead of being straight lines, are Jordan curves (simple and closed). In addition, instead of being limited to just two colors, we allow arbitrarily many colors. In this formulation, essentially arbitrary connected regions can be allocated to agents. It can be interpreted as an abstract model for extreme gerrymandering.

There are $m$ points in the plane, and each point is colored one of $n$ colors. There are $k \geq 2$ agents. The number of points of each color is divisible by $k$. We want to carve the plane into regions using pairwise disjoint Jordan curves so that the resulting regions can be fairly allocated among the agents: each agent should receive the same number of points of each color. We want to minimize the number of curves used.

\begin{lemma}
\label{lem:disks}
If $D_1, \dots, D_i \subset \Reals^2$ are pairwise disjoint closed topological disks, then $\Reals^2 \setminus (D_1 \cup \dots \cup D_i)$ is open and connected.
\end{lemma}

\begin{proof}
The set $D_1 \cup \dots \cup D_i$ is closed, so its complement is open. By the Jordan--Schoenflies theorem~\cite{Newman1951}, each $D_j$ is the image of a round disk under a homeomorphism of $\Reals^2$, so $\Reals^2 \setminus D_j$ is connected. Connectedness of the complement of the union then follows by induction on $i$ from Janiszewski's theorem~\cite{Newman1951}: if $A, B \subset \Reals^2$ are disjoint compact sets whose complements are connected, then the complement of $A \cup B$ is connected.
\end{proof}

\begin{prop}
\label{prop:jordan_curve}
Using pairwise disjoint Jordan curves as cuts, $k-1$ curves are necessary and sufficient.
\end{prop}

\begin{proof}
(Lower bound.) By induction on the number of curves, using the Jordan curve theorem, $i-1$ pairwise disjoint Jordan curves partition the plane into $i$ regions. $k$ agents require at least $k$ regions, so $k-1$ Jordan curves are necessary.

(Upper bound.) Let $m_j$ be the number of points of color $j$. We construct pairwise disjoint closed topological disks $D_1, \dots, D_{k-1}$, where $D_i$ contains exactly $\frac{m_j}{k}$ points of each color $j$. Suppose $D_1, \dots, D_{i-1}$ have been constructed. Select $\frac{m_j}{k}$ points of each color $j$ outside $D_1 \cup \dots \cup D_{i-1}$, and let $P_i$ be the set of remaining unselected points outside these disks. Let
\[U_i = \Reals^2 \setminus (D_1 \cup \dots \cup D_{i-1} \cup P_i)\,.\]
By Lemma~\ref{lem:disks}, and since removing finitely many points preserves connectedness of open sets in the plane, $U_i$ is open and connected. Hence the selected points can be joined by a tree $T_i \subset U_i$ whose edges are polygonal paths meeting only at common endpoints (add the points one at a time, connecting each to the current tree by a simple polygonal path in $U_i$ stopped at its first contact with the tree). Since $T_i$ is compact and $U_i$ is open, $T_i$ has positive distance from $\Reals^2 \setminus U_i$. Let $D_i$ be the closed $\varepsilon$-neighborhood of $T_i$. Then for $\varepsilon$ sufficiently small, $D_i \subset U_i$ and
$D_i$ is a closed topological disk~\cite{Rourke1972}. Finally, allocate the connected region $\Reals^2 \setminus (D_1 \cup \dots \cup D_{k-1})$ to the $k$th agent. The boundaries $\partial D_1, \dots, \partial D_{k-1}$ are $k-1$ pairwise disjoint Jordan curves.
\end{proof}

A visualization of the construction from Proposition~\ref{prop:jordan_curve} is shown in Figure~\ref{fig:jordan_curve}.

\begin{figure}[bhpt]
\centering
\includegraphics[page=12]{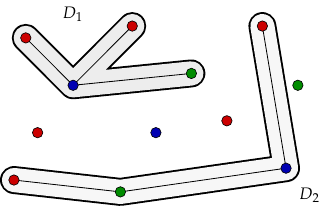}
\caption{The construction of Proposition~\ref{prop:jordan_curve} with $n=3$ colors and $k=3$ agents. $\frac{m_j}{k}$ points of each color are selected and joined by a polygonal tree~$T_1$; $D_1$ is its closed $\varepsilon$-neighborhood. The tree~$T_2$ is built in the complement of $D_1$ and the unselected points, giving~$D_2$. The remaining connected region is allocated to the third agent. The boundaries $\partial D_1$ and $\partial D_2$ are the two Jordan curves.}
\label{fig:jordan_curve}
\end{figure}

\section{Angular Ordering}
\label{sec:angular_ordering}
In this section, we discuss the angular ordering problem. Given a set of input points in general position, fix a viewing point and record the angle at which each input point lies with respect to the viewing point. Sort the input points by their angles; this is their angular ordering. We want to determine the number of angular orderings over all viewing points.
This problem is well studied: D\'{i}az-B\'{a}\~{n}ez et al.~\cite{DiazBanez2014} prove an $O\bigl(m^4\bigr)$ upper bound and Kyn\v{c}l and Soukup~\cite{Kyncl2025} a matching $\Omega\bigl(m^4\bigr)$ lower bound for points in general position, but to the best of our knowledge all existing results are asymptotic, not exact; we present more precise results here.

\begin{remark}
The angular orderings studied here treat the input points as labeled. Each angular ordering corresponds to a cell of the ray arrangement $\mathcal{A}_{\mathrm{ray}}$ of Section~\ref{sec:star_approach}, and hence to one execution of the fairness test, Stage (i), of the star approach; the realization stage (ii) is instead indexed by cells of the finer arrangement $\mathcal{A}_{\mathrm{line}}$. Moreover, for Stage (i) only the induced circular sequence of colors matters, so the relevant quantity is the number of \emph{colored orderings}~\cite{DiazBanez2014}, in which points of the same color are indistinguishable. Every colored ordering is induced by at least one angular ordering, so our counts are upper bounds on the number of distinct necklaces the star approach must test; obtaining exact counts of colored orderings is an interesting direction for future work.
\end{remark}

\subsection{Bounds on number of orderings}
To describe an angular ordering, fix an arbitrary input point as the initial one and describe all angular orderings with it placed first. All rotations of this ordering are considered equivalent, but the one with the initial point placed first is the canonical representation. Thus, there are $(m-1)!$ possible angular orderings (but not all can be geometrically realized).

Draw a line through each pair of input points and then remove the line segments between each pair (leaving only the rays). Within each cell of the resulting partition, all center points map to the same angular ordering. So we only need to try one center point in each cell. The corresponding map of cells to angular orderings is neither injective nor surjective. The size of the range is (combinatorially) bounded above by the number of possible angular orderings: $(m-1)!$.

The size of the range is also (geometrically) bounded by the size of the domain (the set of cells). For $m \geq 7$, this gives us a tighter bound. If we used full lines, we would have $\frac{1}{8} m^4 - \frac{3}{4} m^3 + \frac{23}{8} m^2 - \frac{13}{4} m + 1$ cells (OEIS sequence A055503). Since we are only using rays, we have fewer. If we have convex input, we can compute the exact number as follows: take the number of cells created by the full lines (A055503), subtract the number of cells created by the just the line segments (A006522), add $1$ to account for the resulting central cell, and subtract $m$ to account for the merging of the cell adjacent to each side of the polygon induced by the input points with the central cell. This gives us exactly $\frac{1}{12} m^4 - \frac{1}{2} m^3 + \frac{23}{12} m^2 - \frac{5}{2} m + 1$ cells.

We can also compute the number of cells for nonconvex input as follows. Let $I_{\text{in}}$ (internal) be the number of intersections of line segments (if they weren't removed). Let $I_{\text{ex}}$ (external) be the number of intersections of rays. Add a point at infinity where all the rays end. We can then compute the number of cells using Euler's formula. We have $m + I_{\text{ex}} + 1$ vertices and $m(m-1) + 2I_{\text{ex}}$ edges (each intersection splits two rays into four pieces), so we have $2 - (m + I_{\text{ex}} + 1) + m(m-1) + 2I_{\text{ex}} = m(m-2) + 1 + I_{\text{ex}}$ faces/cells.

Every intersection is associated with a subset of 4 input points (two lines, two input points per line), so by considering all subsets of size 4, we can understand the value of $I_{\text{ex}}$. For any convex subset of size 4, if we draw full lines through each pair of points, there is one internal intersection and two external intersections. For nonconvex subsets, there are no intersections. Thus, if we sum over all subsets of size 4, we have $I_{\text{ex}} = 2I_{\text{in}}$. So we can equivalently write the number of cells as $m(m-2) + 1 + 2I_{\text{in}}$. For convex input, we have $I_{\text{in}} = \binom{m}{4}$, and we recover the original formula for the number of cells for convex input. For nonconvex input, there are always fewer internal intersections, so the number of cells is always smaller. This suggests that the number of angular orderings is likely maximized with convex input.

\subsection{Tightening bounds with equivalent cells}
We now show how we can further tighten the bound on the number of angular orderings by identifying \emph{equivalent} cells --- those that map to the same angular ordering. Each pair of equivalent cells found tightens the bound by one angular ordering; it is known that there can be no set of three cells that map to the same angular ordering~\cite{Kyncl2025}, so we search only for pairs.

\begin{prop}
\label{prop:angular_orderings}
For nearly regular input on a circle (perturbed to ensure general position), when $m \geq 5$, there are at most $\frac{1}{12} m^4 - \frac{1}{2} m^3 + \frac{23}{12} m^2 - \frac{7}{2} m + 1$ angular orderings.
\end{prop}

\begin{proof}
When $m \geq 5$, for each sequence of four consecutive input points $p_i, \dots, p_{i+3}$, the rays $\overrightarrow{p_ip_{i+1}}$ and $\overrightarrow{p_{i+3}p_{i+2}}$ intersect\footnote{When $m < 5$, at least one pair of rays only meets at infinity.}. Each of the $m$ such pairs of rays contributes a tip to an $m$-pointed star polygon that constitutes the central cell of the arrangement. Starting from the central cell, crossing the line $\overleftrightarrow{p_{i+1}p_{i+2}}$ on each side of the removed line segment $\overline{p_{i+1}p_{i+2}}$ leads to a distinct cell that maps to the same angular ordering: the one corresponding to the transposition $(p_{i+1}p_{i+2})$. The $m$ distinct sequences of four consecutive input points correspond to $m$ distinct transpositions, hence $m$ distinct angular orderings.

A set of points on a circle is convex, so there are $\frac{1}{12} m^4 - \frac{1}{2} m^3 + \frac{23}{12} m^2 - \frac{5}{2} m + 1$ cells. We have identified $m$ pairs of equivalent cells, so there are at most $\frac{1}{12} m^4 - \frac{1}{2} m^3 + \frac{23}{12} m^2 - \frac{7}{2} m + 1$ angular orderings.
\end{proof}

For nearly regular input on a circle, both $m=5$ and $6$ have an extra pair of equivalent cells. We computationally determined~\cite{Advani2026b} that there are no extra pairs for $m \in \set{7, \dots, 16}$; it is not known if there are extra pairs for $m \geq 17$.

In an experiment, for $m \leq 9$, we randomly generated input point sets and compared the angular ordering counts to that achieved by nearly regular input on a circle. The only case in which nearly regular input was surpassed was $m=6$: a convex arrangement with 49 corresponding angular orderings was found. It does not, however, improve upon our bound, as $m=6$ has an extra pair of equivalent cells with nearly regular input; achieving 49 simply brings $m=6$ up to par with the other values of $m$.

\begin{conj}
Proposition~\ref{prop:angular_orderings} holds for arbitrary input (when $m \geq 5$).
\end{conj}

\begin{remark}
If two cells map to the same angular ordering, there is a path between them over which the accumulated transpositions compose to the identity. Thus, we are looking for a geometrically realizable set of distinct transpositions whose product is the identity permutation. The $m$ sets we have identified are each of the following form \[\set{(12), (13), \dots, (1(m-1)), (02), (03), \dots, (0(m-1))}\]
Furthermore, any path between the two cells must correspond to an identity permutation, so it is sufficient to analyze the straight-line path between two representative points.
\end{remark}

\section{Time Complexity Calculations}
\label{sec:calculations}
Recall that $k = 2^{e_1} + 2^{e_2} + \dots + 2^{e_\ell}$ with $e_1 > e_2 > \dots > e_\ell$.

\begin{lemma}
If $\ell > 1$, then $\Theta\Bigl(\sum_{j=1}^{\ell-1} \Bigl(\frac{m}{k} \sum_{i=j}^\ell 2^{e_i}\Bigr) \log\Bigl(\sum_{i=j}^\ell 2^{e_i}\Bigr)\Bigr) = \Theta(m \log k)$.
\end{lemma}

\begin{proof}
We have the following.
\begin{align*}
    \sum_{j=1}^{\ell-1} \Bigl(\frac{m}{k} \sum_{i=j}^\ell 2^{e_i}\Bigr) \log\Bigl(\sum_{i=j}^\ell 2^{e_i}\Bigr) &= \frac{m}{k} \sum_{j=1}^{\ell-1} \Bigl(\sum_{i=j}^\ell 2^{e_i}\Bigr) \log\Bigl(\sum_{i=j}^\ell 2^{e_i}\Bigr) \\
    &\geq \frac{m}{k} \Bigl(\sum_{i=1}^\ell 2^{e_i}\Bigr) \log\Bigl(\sum_{i=1}^\ell 2^{e_i}\Bigr) \\
    &= \frac{m}{k} \cdot k \log k \\
    &= m \log k \\
    \intertext{Similarly:}
    \sum_{j=1}^{\ell-1} \Bigl(\frac{m}{k} \sum_{i=j}^\ell 2^{e_i}\Bigr) \log\Bigl(\sum_{i=j}^\ell 2^{e_i}\Bigr) &= \frac{m}{k} \sum_{j=1}^{\ell-1} \Bigl(\sum_{i=j}^\ell 2^{e_i}\Bigr) \log\Bigl(\sum_{i=j}^\ell 2^{e_i}\Bigr) \\
    &\leq \frac{m}{k} \sum_{j=1}^{\ell-1} \Bigl(\sum_{i=0}^{e_j} 2^i\Bigr) \log\Bigl(\sum_{i=1}^\ell 2^{e_i}\Bigr) \\
    &< \frac{m}{k} \sum_{j=1}^{\ell-1} \Bigl(2^{e_j + 1}\Bigr) \log k \\
    &= \frac{2m}{k} \sum_{j=1}^{\ell-1} \Bigl(2^{e_j}\Bigr) \log k \\
    &< \frac{2m}{k} \sum_{j=1}^\ell \Bigl(2^{e_j}\Bigr) \log k \\
    &= \frac{2m}{k} \cdot k \log k \\
    &= 2m \log k
\end{align*}
\end{proof}

\begin{lemma}
$\Theta\bigl(\frac{m}{k} \sum_{i=1}^\ell e_i 2^{e_i}\bigr) = \Theta(m \log k)$.
\end{lemma}

\begin{proof}
We have
\[\frac{m}{k} \sum_{i=1}^\ell e_i 2^{e_i} \geq \frac{m}{k} \cdot e_1 2^{e_1} > \frac{m}{k} \cdot e_1 \frac{k}{2} = \frac{me_1}{2} > \frac{m}{2} \log_2\frac{k}{2} \,,\]
and
\[\frac{m}{k} \sum_{i=1}^\ell e_i 2^{e_i} \leq \frac{me_1}{k} \sum_{i=1}^\ell 2^{e_i} = \frac{me_1}{k} k = me_1 \leq m \log_2 k \,.\]
\end{proof}

\clearpage
\section{Necklace Diagrams}
\label{sec:necklaces}
\begin{figure}[bhpt]
\centering
\includegraphics[width=\linewidth]{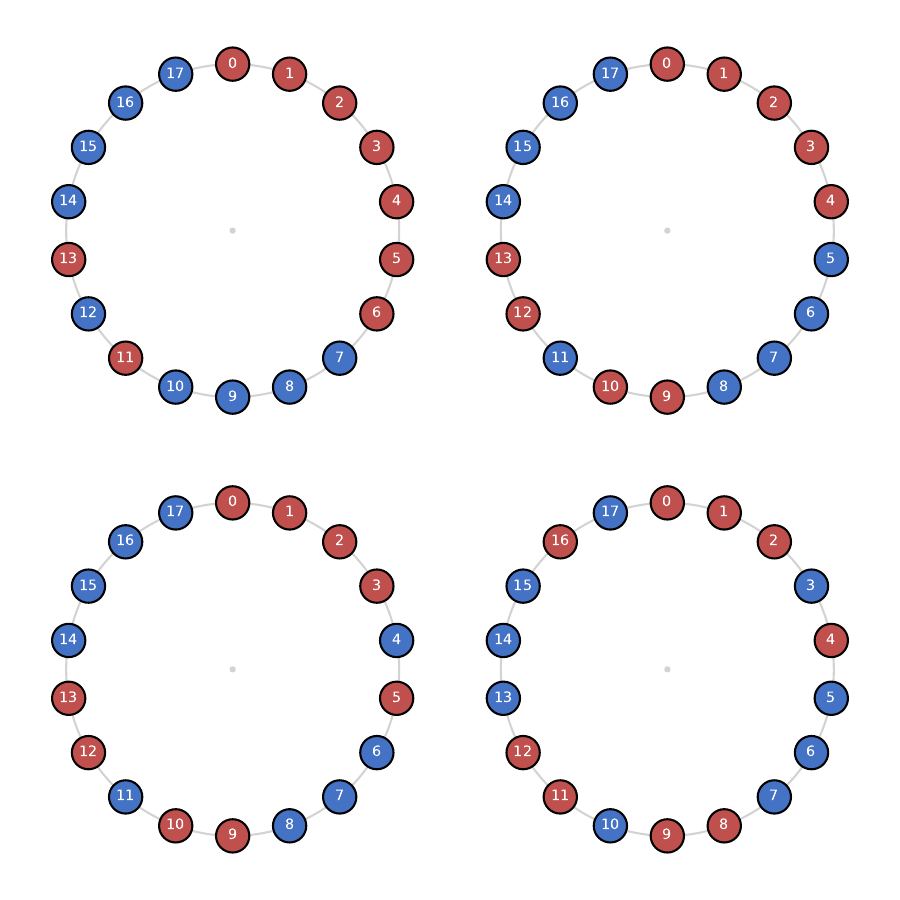}
\caption{All evenly spaced necklaces with $m=18$ beads that cannot be fairly split using fewer than $2(k-1)$ generalized cuts for $k=3$.}
\label{fig:necklaces_k3m18r9}
\end{figure}

\begin{figure}[bhpt]
\centering
\includegraphics[width=\linewidth]{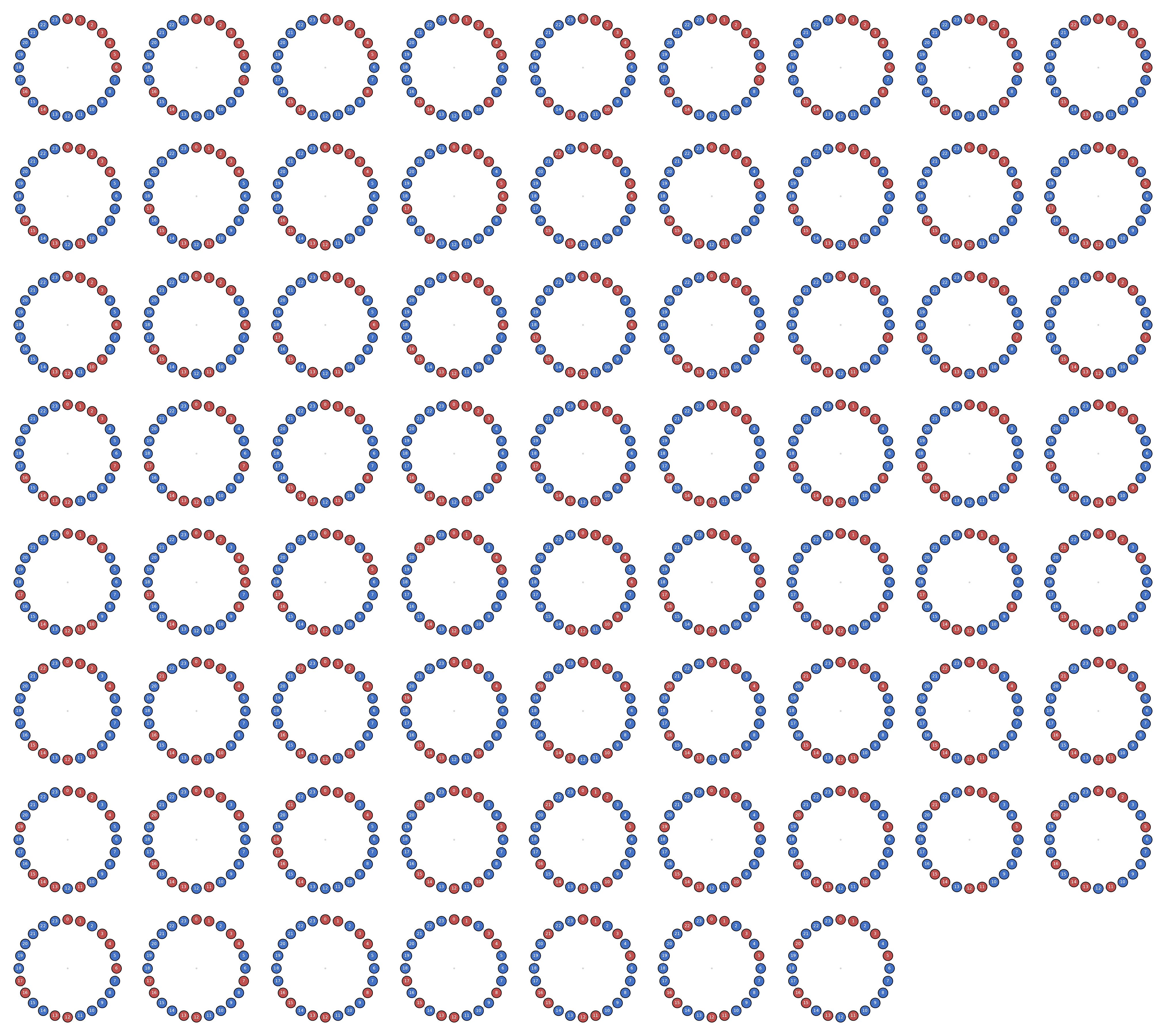}
\caption{All evenly spaced necklaces with $m=24$ beads (and $r=9$) that cannot be fairly split using fewer than $2(k-1)$ generalized cuts for $k=3$.}
\label{fig:necklaces_k3m24r9}
\end{figure}

\begin{figure}[bhpt]
\centering
\includegraphics[width=\linewidth]{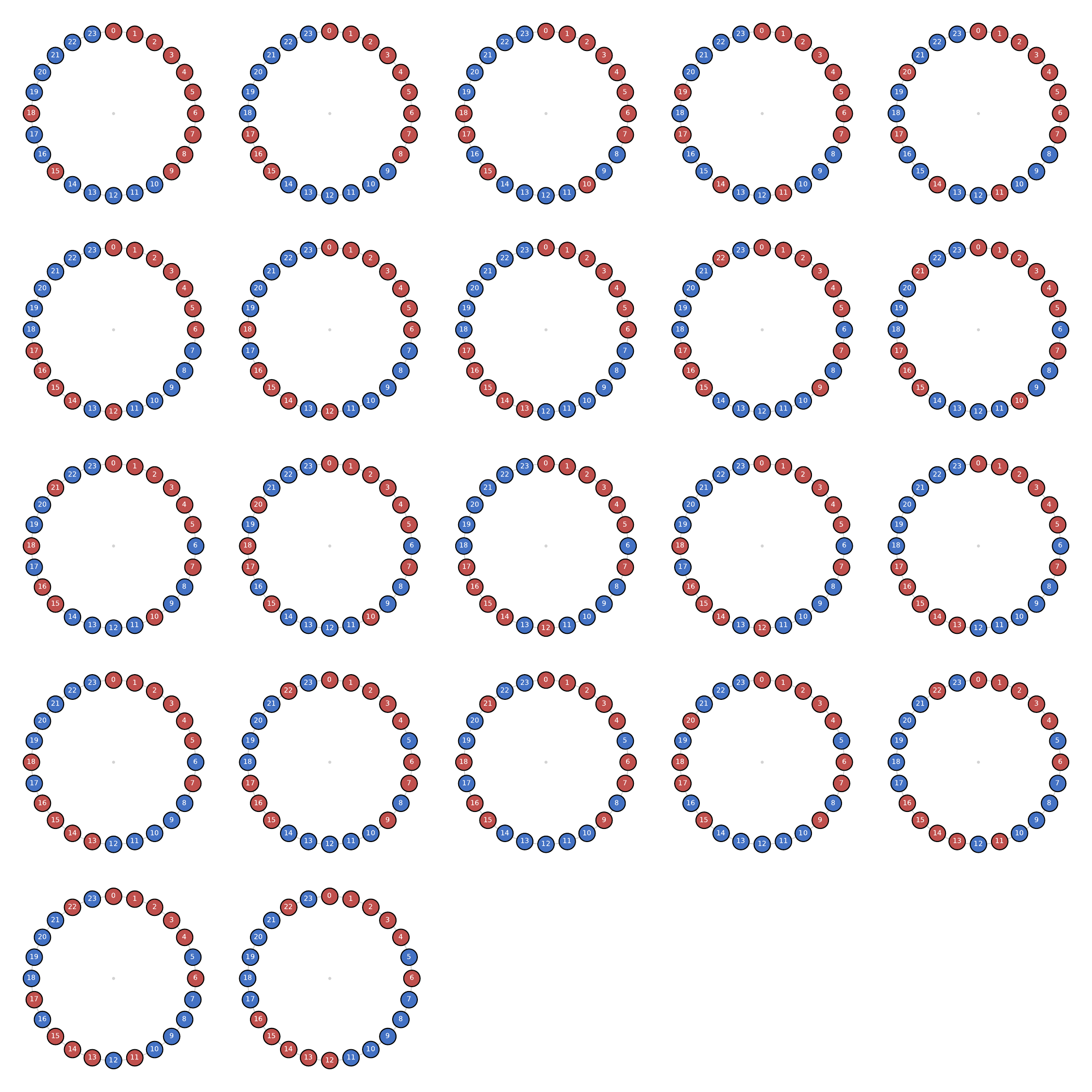}
\caption{All evenly spaced necklaces with $m=24$ beads (and $r=12$) that cannot be fairly split using fewer than $2(k-1)$ generalized cuts for $k=3$.}
\label{fig:necklaces_k3m24r12}
\end{figure}

\end{document}